\documentclass[12pt]{article}
\usepackage[margin=1in]{geometry}
\usepackage{amsmath,amssymb,amsthm}
\usepackage{natbib}
\usepackage{setspace}
\usepackage{booktabs}
\usepackage{graphicx}
\usepackage{hyperref}
\usepackage{caption}
\usepackage{pgfplots}
\pgfplotsset{compat=1.17}
\usepackage{tikz}
\usetikzlibrary{arrows.meta, positioning, calc}

\newtheorem{proposition}{Proposition}
\newtheorem{corollary}{Corollary}

\theoremstyle{definition}
\newtheorem{definition}{Definition}
\newtheorem{assumption}{Assumption}

\begin{document}

\begin{titlepage}
\centering
\vspace*{2.5cm}
{\LARGE \textbf{The Inflation of Resetting Workers}}
\vspace{2.0cm}

{\large Rui Sun\textsuperscript{*}}
\vspace{1.8cm}

\today
\vspace{1.5cm}

\begin{minipage}{0.85\textwidth}
\small
The standard wage Phillips curve aggregates away from which workers reset wages when. I show this aggregation omits a first-order term: the covariance between workers' cost-push exposure and their reset frequency. I introduce two sufficient statistics and embed them in a multi-country HANK model calibrated to six euro-area economies. The omitted term generates 7 percent more cumulative core inflation in the baseline and 10--26 percent more when monetary policy is delayed. Two economies with identical openness can differ by 6.6 percentage-point-quarters solely from within-country composition. Targeted essentials subsidies reduce welfare loss by 32 percent relative to aggressive tightening. Out of sample, the model correctly predicts the persistence ranking across the UK, the US, and Japan.

\vspace{0.5cm}
\end{minipage}

\vfill
\begin{flushleft}
\rule{0.4\textwidth}{0.4pt}\\[6pt]
{\footnotesize \textsuperscript{*}Haas School of Business, University of California, Berkeley (rui\_sun@berkeley.edu). I thank Yuriy Gorodnichenko, Emi Nakamura, J\'{o}n Steinsson, Pierre-Olivier Gourinchas, Benjamin Schoefer, and Dmitry Mukhin for guidance. I also thank seminar participants at UC Berkeley for helpful comments. All errors are my own.}
\end{flushleft}
\end{titlepage}

\section{Introduction}\label{sec:intro}

The standard wage Phillips curve uses the wrong price index for the workers who set wages. In the \citet{erceg2000} framework, extended by \citet{blanchard2007} to open economies, all workers target the aggregate CPI when resetting their nominal wage. But the workers who reset most frequently are not a random sample of the population: in European labor markets, they are disproportionately low-income, concentrated in retail, hospitality, and personal services, operating under contracts of one to two years \citep{buchheim2022}. These same workers spend 33--40 percent of their budget on food and energy, compared with 17--22 percent for workers who reset less often. When an imported cost-push shock hits essentials prices, the workers who experience the most inflation are precisely the workers in a position to renegotiate wages. The standard model averages over this composition and misses the resulting wage pressure. Throughout the paper, ``the standard model'' refers to the nested special case in which all workers share a common consumption basket and a common reset frequency, so that $\Omega_{c,t}=0$ identically; it is the representative-agent Calvo wage Phillips curve of \citet{erceg2000}.

I derive a heterogeneous wage Phillips curve in which each worker type targets its own price index. The aggregate contains a term absent from the standard model: the within-country covariance between a worker's wage-reset frequency and the inflation that worker experiences. I call this the reset-heterogeneity wedge and denote it $\Omega_{c,t}$. The wedge is identically zero only under knife-edge conditions, namely identical consumption baskets, identical reset frequencies, or uniform price changes, that the representative-agent model implicitly imposes. The 2021--2023 euro-area inflation episode, in which energy and food prices surged by 40--60 percent while core inflation diverged sharply across countries sharing the same central bank, is the natural laboratory: cumulative core inflation overshoot ranged from 31 percentage-point-quarters in France to 51 in the Netherlands, a dispersion that standard models, which attribute cross-country differences to trade openness or the output gap, cannot account for since the countries at the extremes have comparable aggregate openness and labor market tightness.

The central objects in the analysis are two sufficient statistics. Reset-Weighted Experienced Inflation (RWEI) measures the average inflation experienced by workers at the reset margin; it predicts wage catch-up. Marginal Wage Setter Inflation (MWSI) augments RWEI with sectoral propagation weights that capture how wage increases in each sector feed into downstream prices through the input-output network; it predicts cumulative core inflation persistence. Both objects are computable from three publicly available data sources, namely household expenditure surveys, wage-setting calendars, and input-output tables, for any country and any shock composition, without solving any equilibrium model.

The main contribution is the identification and characterization of the composition channel in the cumulative wage response. The cumulative effect of a cost-push shock decomposes into three additively separable channels:
\begin{equation}\label{eq:decomp_channels}
\underbrace{\sum_h \beta^h \pi^w_{c,h}}_{\text{total}} = \underbrace{\frac{\bar\theta_c}{1-\beta}\bar\pi_c(u_0)}_{\text{level catch-up}} + \underbrace{\frac{\bar\theta_c}{1-\beta\rho_\Omega}\Omega_{c,0}}_{\substack{\text{composition wedge}\\\text{(new)}}} + \underbrace{\frac{\tilde\kappa_c}{1-\beta}\bar x}_{\text{demand}}\,.
\end{equation}
The first channel operates through average inflation; the second through the cross-sectional wedge $\Omega$. The composition channel is new. The wedge is a tail-risk amplifier: it is proportional to the shock size and therefore negligible for small shocks but first-order during the large supply disruptions that matter most for policy. Critically, the aggregation error is shock-dependent: the standard Phillips curve works well for demand shocks, which raise prices broadly, but fails for supply shocks concentrated on necessities. No recalibration of the representative-agent model can replicate this asymmetry. The same mechanism runs in reverse: when oil prices collapsed in 2014--2016, the wedge turned negative, and the standard model overpredicted wage disinflation. One object, two puzzles, opposite signs.

I embed the new Phillips curve in a multi-country heterogeneous-agent New Keynesian model calibrated to six euro-area economies: Germany, France, Italy, Spain, the Netherlands, and Belgium. The quantitative model uses five income quintiles to capture the full expenditure gradient observed in the Eurostat Household Budget Survey. A key feature of the quantitative model is asymmetric wage catch-up: workers who have fallen behind in real terms demand larger raises at the next reset, but workers who are ahead do not volunteer pay cuts. This asymmetry, documented by \citet{buchheim2022} and visible in every major European wage negotiation during 2022--2023, makes the wedge self-reinforcing during large supply shocks while leaving it dormant during small ones. I solve the model using the sequence-space Jacobian method of \citet{auclert2021} and extend it to nonlinear transitions via the quasi-Newton method described in their Section~6. Using pre-shock micro data on expenditure shares, contract durations, and sectoral structure, and without targeting any inflation outcome from the 2021--2023 episode, the model shows that the standard economy underestimates cumulative core inflation by 7 percent in the baseline and by 10--26 percent when monetary policy is delayed, as the ECB's was during 2021--2022. Two economies with identical aggregate openness and identical average wage-setting institutions can differ in cumulative core inflation by 6.6 percentage-point-quarters solely because their within-country distributions of expenditure and reset timing differ.

The wedge has direct implications for stabilization policy. Because it is a cross-sectional object that depends on which workers experience which inflation, the aggregate interest rate cannot close it. This is a new form of the divine coincidence failure. In the monetary union, cross-country variation in RWEI creates a new dimension of the optimal currency area problem \citep{mundell1961}: no common interest rate can close all country-specific wedges, and the direction of bias, specifically which countries are over-tightened and which are under-tightened, is predicted by RWEI. Targeted fiscal transfers to workers at the reset margin are net disinflationary at the medium horizon: moderate tightening combined with essentials subsidies reduces the union-wide welfare loss by 32 percent relative to aggressive tightening alone, because the subsidy directly weakens the wage catch-up channel that fuels second-round inflation. As an out-of-sample validation, the model, calibrated entirely to the euro area, correctly predicts that the United Kingdom should exhibit more persistent core inflation than the United States, and the United States more than Japan, using only pre-shock expenditure and contract data from each country.

This paper contributes to several literatures. First, I build on the theory of wage and price Phillips curves. The standard framework \citep{calvo1983,erceg2000} uses a representative worker whose experienced inflation equals the aggregate CPI. Multi-sector models \citep{carvalho2006,nakamura2010} introduce heterogeneous price stickiness across products but maintain a representative worker within each sector. \citet{hazell2022} estimate slope heterogeneity in regional Phillips curves. I introduce heterogeneity across worker types, in both consumption baskets and wage-reset timing, which generates a new aggregate object absent from all three frameworks and provides a structural explanation for why different compositions of workers can generate different Phillips curve slopes even with the same aggregate conditions.

The closest paper is \citet{afrouzi2024}, who model how workers keep up with inflation in a closed economy with a representative consumption basket. The two papers are complements: \citeauthor{afrouzi2024} answer why workers catch up, the level channel in equation~(\ref{eq:decomp_channels}), while I answer which workers' catch-up matters for aggregate dynamics, the composition channel. In their framework $\Omega=0$ identically; the wedge is the new margin. I also build on the literature on inflation incidence and experienced prices. \citet{pallotti2023} measure heterogeneous welfare losses from the 2021--2023 euro-area inflation; I show that incidence has aggregate consequences beyond welfare. \citet{cavallo2017} and \citet{dacunto2021} document that experienced prices shape household expectations; I use their estimates to discipline the expectations block but show that heterogeneous exposure, not heterogeneous beliefs, is the primary driver.

The paper is also related to recent work on supply-side inflation dynamics. \citet{acharya2023} document how supply-chain pressures generalized inflation across European sectors, \citet{afrouzi2023rp} show that relative-price shocks can act as aggregate supply shocks, and \citet{blanco2024} study inflation dynamics in menu-cost economies. \citet{bernanke2025} find that pandemic-era US inflation was driven by supply shocks rather than overheated labor markets. I provide a sufficient statistic for which relative-price shocks generate the most persistence, and show that the persistence of those supply shocks depends on which workers they hit, a margin that representative-agent supply-side frameworks cannot capture.

On the open-economy side, the canonical model of \citet{gali2005} features a representative agent within each country. \citet{digiovanni2023} develop a multi-country multi-sector model to study pandemic-era inflation through global production networks, but their model maintains representative workers within each country-sector and therefore cannot generate the within-country composition effects that drive my results. \citet{cubaborda2023} show that intermediate-input trade-cost shocks are more persistent than final-goods shocks; \citet{gopinath2020} emphasize dominant-currency pricing. The present mechanism operates through the domestic labor market and is complementary to these trade channels.

I contribute to the theory of optimal currency areas by adding a new criterion to \citet{mundell1961} and \citet{degrauwe2020}: RWEI dispersion captures within-country composition rather than between-country asymmetry. \citet{kekre2022} shows that between-country heterogeneity in labor market frictions creates OCA costs even with symmetric shocks; I identify a complementary source of asymmetry that operates within countries, through the composition of who resets wages and what inflation they experience, rather than through between-country differences in hiring costs or matching efficiency. \citet{hauserseneca2022} study labor mobility as an adjustment mechanism in a monetary union; the present criterion identifies a source of asymmetry that labor mobility cannot resolve, since the relevant heterogeneity is within each country's labor market, not between countries.

Finally, the HANK literature following \citet{kaplan2018} and \citet{mckay2016} is overwhelmingly closed-economy and demand-focused; this paper brings household heterogeneity to the supply side, a channel orthogonal to the MPC channel. The computational method follows \citet{auclert2021}; for alternative approaches, see \citet{bayer2020}.

The rest of the paper is organized as follows. Section~\ref{sec:facts} presents motivating evidence. Section~\ref{sec:model} sets up the model. Section~\ref{sec:theory} derives the main theoretical results. Section~\ref{sec:quant} describes the quantitative model, calibration, and computation. Section~\ref{sec:results} reports quantitative results and policy analysis. Section~\ref{sec:conclusion} concludes.

\section{Motivating Evidence}\label{sec:facts}

The euro area offers a laboratory in which member countries share a single central bank, a common currency, and broadly similar external price shocks, yet exhibit heterogeneous inflation dynamics. This section documents three patterns.

\subsection{Cross-country dispersion in core inflation persistence}

Figure~\ref{fig:core} plots core HICP inflation excluding energy, food, alcohol, and tobacco for five large euro-area economies from 2021 through 2024. All countries experienced rising core inflation during 2022--2023, but the peak levels and subsequent rates of decline differed.

\begin{figure}[htbp]
\centering
\begin{tikzpicture}
\begin{axis}[
    width=13.5cm, height=7cm,
    xlabel={Quarter},
    ylabel={Core HICP (\%, y-o-y)},
    xmin=0.5, xmax=16.5, ymin=-0.5, ymax=8.5,
    xtick={1,4,8,12,16},
    xticklabels={2021:1, 2021:4, 2022:4, 2023:4, 2024:4},
    legend style={at={(0.02,0.98)}, anchor=north west, font=\footnotesize, draw=none},
    grid=major, grid style={dashed, gray!25},
    every axis plot/.append style={thick}
]
\addplot[blue, mark=square*, mark size=1.2] coordinates {
(1,1.4)(2,1.8)(3,2.4)(4,3.2)(5,3.4)(6,3.8)(7,4.3)(8,5.1)
(9,5.5)(10,5.7)(11,5.3)(12,4.8)(13,4.0)(14,3.5)(15,3.0)(16,2.8)};
\addplot[red, mark=triangle*, mark size=1.2] coordinates {
(1,0.8)(2,1.1)(3,1.5)(4,2.0)(5,2.4)(6,3.1)(7,3.8)(8,4.2)
(9,4.5)(10,4.7)(11,4.3)(12,3.6)(13,3.0)(14,2.5)(15,2.1)(16,1.8)};
\addplot[green!60!black, mark=diamond*, mark size=1.2] coordinates {
(1,0.6)(2,0.9)(3,1.6)(4,2.5)(5,3.2)(6,4.5)(7,5.4)(8,5.8)
(9,6.3)(10,6.1)(11,5.2)(12,4.0)(13,3.1)(14,2.3)(15,1.8)(16,1.6)};
\addplot[orange, mark=o, mark size=1.2] coordinates {
(1,0.4)(2,0.8)(3,1.4)(4,2.1)(5,3.0)(6,4.2)(7,5.0)(8,5.8)
(9,6.1)(10,5.6)(11,4.8)(12,4.1)(13,3.5)(14,3.2)(15,2.9)(16,2.6)};
\addplot[purple, mark=star, mark size=1.5] coordinates {
(1,1.2)(2,1.8)(3,2.6)(4,3.8)(5,4.2)(6,5.0)(7,6.2)(8,7.2)
(9,8.0)(10,7.3)(11,6.0)(12,4.8)(13,3.8)(14,3.2)(15,2.7)(16,2.4)};
\legend{Germany, France, Italy, Spain, Netherlands}
\end{axis}
\end{tikzpicture}
\caption{Core HICP inflation for selected euro-area countries.}\label{fig:core}
\begin{flushleft}
\footnotesize Note: Source: Eurostat HICP (prc\_hicp\_manr), series ``All-items HICP excluding energy, food, alcohol, and tobacco,'' quarterly averages of monthly year-on-year rates, 2021:1--2024:4.
\end{flushleft}
\end{figure}
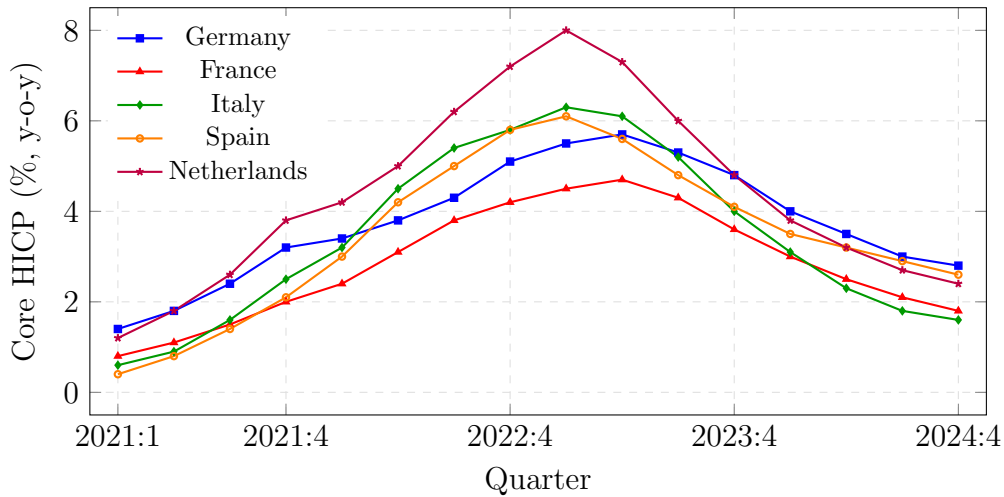

To measure persistence, I define the cumulative core inflation overshoot in country~$c$ as $\Pi^{core}_c \equiv \sum_{t} (\pi^{core}_{c,t} - 2\%)$ over 2021:3--2024:4. Table~\ref{tab:persistence} reports the results. The cross-country standard deviation of $\Pi^{core}_c$ is 6.8 percentage-point-quarters.

\begin{table}[htbp]
\centering\small
\caption{Cumulative core inflation overshoot, 2021:3--2024:4}\label{tab:persistence}
\begin{tabular}{lccc}
\toprule
Country & Peak core (\%) & $\Pi^{core}_c$ (pp$\cdot$q) & Peak services (\%) \\
\midrule
Germany & 5.7 & 42.3 & 5.2 \\
France & 4.7 & 30.8 & 4.0 \\
Italy & 6.3 & 38.5 & 4.8 \\
Spain & 6.1 & 42.8 & 5.9 \\
Netherlands & 8.0 & 51.2 & 7.6 \\
Belgium & 6.9 & 46.1 & 6.3 \\
\bottomrule
\end{tabular}

\medskip
\begin{flushleft}
\footnotesize Note: Cumulative overshoot $\Pi^{core}_c\equiv\sum_t(\pi^{core}_{c,t}-2\%)$ in percentage-point-quarters, summed over 2021:3--2024:4. Core HICP excludes energy, food, alcohol, and tobacco. Peak core and peak services are maximum quarterly year-on-year rates. Source: Eurostat.
\end{flushleft}
\end{table}

\subsection{Correlation with within-country labor market composition}

I construct a preliminary measure of RWEI for each country using expenditure shares by income quintile from the 2020 Eurostat Household Budget Survey, item-level HICP inflation at the onset of the episode, and reset weights approximated from the ECB wage tracker and OECD/AIAS ICTWSS database. Table~\ref{tab:expenditure} reports the underlying expenditure data.

\begin{table}[htbp]
\centering\small
\caption{Expenditure shares and contract durations}\label{tab:expenditure}
\begin{tabular}{lcccc}
\toprule
& \multicolumn{2}{c}{Food + energy share (\%)} & Gap & Avg.\ contract \\
\cmidrule(lr){2-3}
Country & Q1 (bottom) & Q5 (top) & (pp) & (years) \\
\midrule
Germany & 36.2 & 19.8 & 16.4 & 2.1 \\
France & 34.5 & 20.4 & 14.1 & 2.5 \\
Italy & 38.1 & 21.7 & 16.4 & 2.8 \\
Spain & 39.8 & 19.1 & 20.7 & 1.8 \\
Netherlands & 33.0 & 18.5 & 14.5 & 2.0 \\
Belgium & 35.8 & 21.2 & 14.6 & 2.3 \\
\bottomrule
\end{tabular}

\medskip
\begin{flushleft}
\footnotesize Note: Food + energy shares are budget shares of COICOP categories CP01 (food) and CP04.5/CP07.2 (energy) by income quintile, from the 2020 Eurostat Household Budget Survey. Gap is Q1 minus Q5 share in percentage points. Average contract duration from OECD/AIAS ICTWSS database and ECB wage tracker.
\end{flushleft}
\end{table}

Figure~\ref{fig:scatter} plots RWEI, as defined in equation~(\ref{eq:rwei}), against cumulative core inflation. The cross-country correlation is weakly positive: RWEI predicts which countries have the largest reset-heterogeneity wedge, but total cumulative core inflation depends on many additional factors such as indexation, bargaining institutions, and slack that RWEI does not capture. The model's value is not in predicting the level of core inflation, which the standard model already does, but in predicting the additional wage pressure from the wedge channel.

\begin{figure}[htbp]
\centering
\begin{tikzpicture}
\begin{axis}[
    width=10cm, height=7cm,
    xlabel={RWEI$_{c,t_0}$ (\%)}, ylabel={$\Pi^{core}_c$ (pp$\cdot$q)},
    xmin=3.5, xmax=7.5, ymin=25, ymax=55,
    grid=major, grid style={dashed, gray!25},
    nodes near coords, every node near coord/.append style={font=\footnotesize, anchor=south west},
    point meta=explicit symbolic
]
\addplot[only marks, mark=*, mark size=2.5, blue!70!black] coordinates {
(4.8,42.3)[DE] (4.2,30.8)[FR] (5.3,38.5)[IT] (6.5,42.8)[ES] (5.8,51.2)[NL] (5.5,46.1)[BE]};
\addplot[red!70!black, dashed, domain=3.5:7.5, samples=2] {11.2+5.8*x};
\end{axis}
\end{tikzpicture}
\caption{RWEI versus cumulative core inflation.}\label{fig:scatter}
\begin{flushleft}
\footnotesize Note: Each point is one euro-area country. Horizontal axis: Reset-Weighted Experienced Inflation (Definition 1), computed from pre-shock expenditure shares and item-level HICP at 2021:3. Vertical axis: cumulative core inflation overshoot $\Pi^{core}_c$ from Table 1. The correlation captures the raw association; the model's value is in explaining the mechanism, not the cross-sectional fit.
\end{flushleft}
\end{figure}
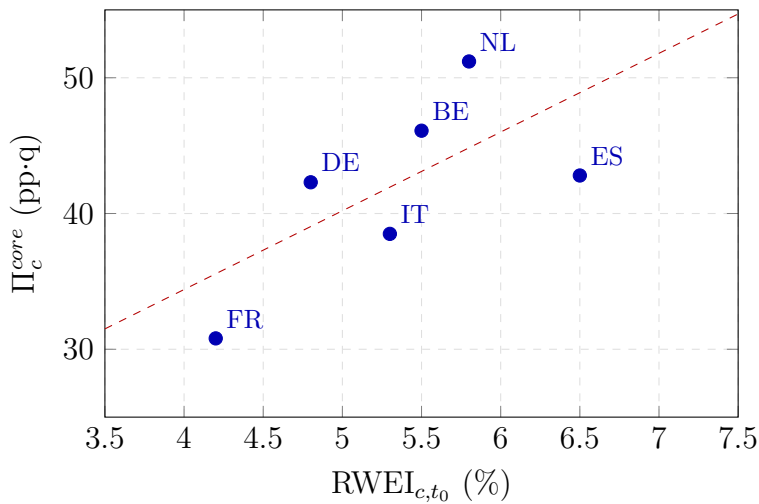

\subsection{The standard model underpredicts dispersion}

For each country, I calibrate a one-sector New Keynesian Phillips curve using country-specific average contract duration and import openness. The standard model generates a cross-country standard deviation of $\Pi^{core}_c$ that is substantially smaller than the observed 6.8 pp$\cdot$q, because it imposes homogeneous baskets and reset timing. The heterogeneous-reset model developed below generates dispersion closer to the data through the wedge channel.

\section{Model}\label{sec:model}

This section specifies the primitive environment used to derive the theoretical results in Section~\ref{sec:theory}. Section~\ref{sec:quant} describes the quantitative extension with idiosyncratic risk and incomplete markets.

\subsection{Overview and timing}

Time is discrete and indexed by $t = 0, 1, 2, \ldots$ The economy is a monetary union of $N$ countries indexed by $c \in \{1,\ldots,N\}$. Within each country there are heterogeneous workers, domestic firms operating in two sectors, and an imported-essentials sector. Countries share a common central bank. Each country has a fiscal authority.

The timing within each period is as follows. At the beginning of period $t$: (i)~the cost-push shock $u_t$ is realized and the imported-essentials price is determined; (ii)~a randomly selected fraction $\theta_{g,c}$ of type-$g$ workers in country $c$ receive the opportunity to reset their nominal wage; (iii)~production takes place, firms set prices, and goods and labor markets clear; (iv)~the central bank sets the nominal interest rate; (v)~consumption and saving occur.

\subsection{Households}\label{sec:hh}

Within each country $c$, there is a unit mass of infinitely lived workers. Workers are divided into two types, $g \in \{H, L\}$, with fixed population shares $\eta_{H,c}$ and $\eta_{L,c} = 1 - \eta_{H,c}$. Type is permanent and publicly observable.

\noindent\textbf{Preferences.} A type-$g$ worker in country $c$ has period utility
\begin{equation}\label{eq:utility}
    U(C_{g,c,t}, N_{g,c,t}) = \frac{C_{g,c,t}^{1-\sigma}}{1-\sigma} - \chi \frac{N_{g,c,t}^{1+\varphi_n}}{1+\varphi_n}\,,
\end{equation}
where $C_{g,c,t}$ is consumption of a composite good and $N_{g,c,t}$ is labor supply. The parameter $\sigma > 0$ is the coefficient of relative risk aversion, $\varphi_n > 0$ is the inverse Frisch elasticity of labor supply, and $\chi > 0$ scales the disutility of labor. Workers discount the future at rate $\beta \in (0,1)$ and maximize $\mathbb{E}_0 \sum_{t=0}^{\infty} \beta^t U(C_{g,c,t}, N_{g,c,t})$.

\noindent\textbf{Consumption basket.} The composite consumption good $C_{g,c,t}$ is a Cobb-Douglas aggregate of three categories: imported essentials ($e$), domestically produced goods ($d$), and services ($s$):
\begin{equation}\label{eq:basket}
    C_{g,c,t} = \left(\frac{C^{e}_{g,c,t}}{\alpha_{g,c,e}}\right)^{\!\alpha_{g,c,e}} \!\cdot\! \left(\frac{C^{d}_{g,c,t}}{\alpha_{g,c,d}}\right)^{\!\alpha_{g,c,d}} \!\cdot\! \left(\frac{C^{s}_{g,c,t}}{\alpha_{g,c,s}}\right)^{\!\alpha_{g,c,s}},
\end{equation}
where $\alpha_{g,c,e} + \alpha_{g,c,d} + \alpha_{g,c,s} = 1$. The expenditure shares $(\alpha_{g,c,e}, \alpha_{g,c,d}, \alpha_{g,c,s})$ are type-specific and country-specific. This is the key heterogeneity in the model: different types of workers allocate different fractions of their budget to different goods.

The ideal price index associated with (\ref{eq:basket}) is
\begin{equation}\label{eq:priceindex}
    P_{g,c,t} = P_{e,t}^{\alpha_{g,c,e}} \cdot P_{d,c,t}^{\alpha_{g,c,d}} \cdot P_{s,c,t}^{\alpha_{g,c,s}}\,,
\end{equation}
where $P_{e,t}$ is the price of imported essentials, common across the union, $P_{d,c,t}$ is the price of domestic goods in country $c$, and $P_{s,c,t}$ is the price of services in country $c$. In logs, the price index is
\begin{equation}\label{eq:pindex}
    p_{g,c,t} = \alpha_{g,c,e}\,p_{e,t} + \alpha_{g,c,d}\,p_{d,c,t} + \alpha_{g,c,s}\,p_{s,c,t}\,.
\end{equation}
Intratemporal optimization implies that expenditure on each category equals its share times total expenditure: $P_{i,c,t} C^i_{g,c,t} = \alpha_{g,c,i} \cdot P_{g,c,t} C_{g,c,t}$ for $i \in \{e,d,s\}$.

\begin{assumption}\label{ass:baskets}
Type-$H$ workers spend a larger share on imported essentials than type-$L$ workers:
\begin{equation*}
    \alpha_{H,c,e} > \alpha_{L,c,e} \qquad \text{for all } c.
\end{equation*}
\end{assumption}

This assumption captures the well-documented empirical regularity that low-income households allocate a larger fraction of their budget to necessities such as food and energy. Across the euro area, households in the bottom income quintile spend 33--40 percent of their budget on food and energy, compared with 17--22 percent for the top quintile (Table~\ref{tab:expenditure}).

\noindent\textbf{Budget constraint.} In the simple model, workers have access to a risk-free nominal bond $B_{g,c,t}$ that pays gross nominal return $1+i_t$, where $i_t$ is the common nominal interest rate set by the union's central bank. The flow budget constraint is
\begin{equation}\label{eq:budget}
    P_{g,c,t}\,C_{g,c,t} + B_{g,c,t+1} = W_{g,c,t}\,N_{g,c,t} + (1+i_{t-1})\,B_{g,c,t} + T_{g,c,t} - \mathcal{T}_{c,t}\,,
\end{equation}
where $W_{g,c,t}$ is the nominal wage of type $g$, $T_{g,c,t}$ is a targeted transfer from the fiscal authority, which is zero in the baseline, and $\mathcal{T}_{c,t}$ is a lump-sum tax levied on all workers. The quantitative extension in Section~\ref{sec:quant} adds idiosyncratic productivity risk and a borrowing constraint.

\subsection{Wage setting}\label{sec:wages}

Nominal wages are set in a Calvo-staggered fashion, following \citet{erceg2000}. In each period, each type-$g$ worker in country $c$ receives, independently and with probability $\theta_{g,c} \in (0,1)$, the opportunity to choose a new nominal wage. With complementary probability $1 - \theta_{g,c}$, the worker retains the nominal wage from the previous period.

\begin{assumption}\label{ass:reset}
Type-$H$ workers reset wages more frequently than type-$L$ workers:
\begin{equation*}
    \theta_{H,c} > \theta_{L,c} \qquad \text{for all } c.
\end{equation*}
\end{assumption}

In European labor markets, workers in services and retail typically operate under contracts of one to two years, while workers in manufacturing and the public sector have contracts lasting two to four years. Since type-$H$ workers are concentrated in the former sectors and type-$L$ workers in the latter, Assumption~\ref{ass:reset} is the natural counterpart of Assumption~\ref{ass:baskets}. The conjunction of these two assumptions, that the workers most exposed to imported cost-push shocks are also the workers who reset wages most frequently, is the source of the mechanism studied in this paper.

\noindent\textbf{Optimal reset wage.} When a type-$g$ worker in country $c$ receives the opportunity to reset at date $t$, the worker chooses the nominal wage $W^*_{g,c,t}$ to maximize the expected discounted utility from the current date until the next reset opportunity, taking into account the probability of retaining the chosen wage in future periods. Formally, substituting consumption from the budget constraint (\ref{eq:budget}) and using the period utility (\ref{eq:utility}), $W^*_{g,c,t}$ solves
\begin{equation}\label{eq:wage_problem}
    \max_{W^*} \sum_{k=0}^{\infty} (\beta\Theta_{g,c})^k\,\mathbb{E}_t\!\left[U\!\left(\frac{W^* N_{g,c,t+k} + (1+i_{t+k-1})B_{g,c,t+k} + T_{g,c,t+k} - \mathcal{T}_{c,t+k}}{P_{g,c,t+k}},\, N_{g,c,t+k}\right)\right],
\end{equation}
where $\Theta_{g,c} \equiv 1 - \theta_{g,c}$ is the probability of not resetting, and labor demand $N_{g,c,t+k}$ depends on $W^*$ through the firm's cost minimization (described below in Section~\ref{sec:firms}). The factor $(\beta\Theta_{g,c})^k$ reflects both discounting and the probability that the wage $W^*$ is still in effect at date $t+k$.

Log-linearizing the first-order condition of (\ref{eq:wage_problem}) around a zero-inflation steady state yields the optimal log reset wage:
\begin{equation}\label{eq:wreset}
    w^*_{g,c,t} = (1-\beta\Theta_{g,c})\sum_{k=0}^{\infty}(\beta\Theta_{g,c})^k\,\mathbb{E}_t\!\left[p_{g,c,t+k} + \mathrm{mrpl}_{g,c,t+k}\right].
\end{equation}
Here $p_{g,c,t+k} = \log P_{g,c,t+k}$ is the type-specific log price index defined in (\ref{eq:pindex}), and $\mathrm{mrpl}_{g,c,t+k}$ is the log marginal revenue product of type-$g$ labor, derived from the firm's problem in Section~\ref{sec:firms}. The derivation follows the standard steps in \citet{erceg2000} and is given in Appendix~\ref{app:wreset}. The key departure from the standard model is that each worker type targets its own price index $p_{g,c,t+k}$ rather than the aggregate CPI: when the imported-essentials price rises, type-$H$ workers, who spend a larger share on essentials, perceive a larger increase in their cost of living and demand commensurately higher wages.

\noindent\textbf{Wage aggregation.} The aggregate nominal wage index for type $g$ in country $c$ evolves as
\begin{equation}\label{eq:wage_index}
    W_{g,c,t} = \left[\theta_{g,c}(W^*_{g,c,t})^{1-\varepsilon_w} + (1-\theta_{g,c})(W_{g,c,t-1})^{1-\varepsilon_w}\right]^{\frac{1}{1-\varepsilon_w}}\,,
\end{equation}
where $\varepsilon_w > 1$ is the elasticity of substitution across differentiated labor services within each type. In log-linearized form, this becomes
\begin{equation}\label{eq:wage_loglin}
    w_{g,c,t} = \theta_{g,c}\,w^*_{g,c,t} + (1-\theta_{g,c})\,w_{g,c,t-1}\,,
\end{equation}
so type-$g$ wage inflation is
\begin{equation}\label{eq:piw_type}
    \pi^w_{g,c,t} \equiv w_{g,c,t} - w_{g,c,t-1} = \theta_{g,c}(w^*_{g,c,t} - w_{g,c,t-1}).
\end{equation}
Aggregate wage inflation in country $c$ is the employment-weighted average:
\begin{equation}\label{eq:piw_agg}
    \pi^w_{c,t} = \sum_{g \in \{H,L\}} \eta_{g,c}\,\pi^w_{g,c,t}\,.
\end{equation}

\subsection{Inflation expectations}\label{sec:expectations}

Workers form expectations about future inflation based on their own experienced inflation. I adopt the specification
\begin{equation}\label{eq:expect}
    \mathbb{E}^g_{c,t}[\pi_{c,t+1}] = \bar{\pi}_{c,t} + \varphi_{g,c}\,\tilde{\pi}^{exp}_{g,c,t}\,,
\end{equation}
where $\bar{\pi}_{c,t}$ is a common signal, for example the central bank's inflation target or recent headline inflation, $\varphi_{g,c} \geq 0$ is a sensitivity parameter, and $\tilde{\pi}^{exp}_{g,c,t}$ is the salient experienced inflation of type $g$, defined as
\begin{equation}\label{eq:piexp}
    \tilde{\pi}^{exp}_{g,c,t} \equiv \sum_{i \in \{e,d,s\}} \lambda_i\,\alpha_{g,c,i}\,\Delta p_{i,c,t}\,.
\end{equation}
Here $\Delta p_{i,c,t} \equiv p_{i,c,t} - p_{i,c,t-1}$ is the log price change in category $i$, and $\lambda_i \geq 1$ is a salience weight that allows frequently purchased items such as food and energy to influence expectations more than their expenditure shares alone would imply.

The parameters $\varphi_{g,c}$ and $\lambda_i$ are not estimated within the model; they are pinned from external micro evidence. The baseline sets $\lambda_e = 1.3$ for essentials and $\lambda_i = 1$ for other goods, so that food and energy are 30 percent more salient than their expenditure share alone would imply, following the estimates in \citet{dacunto2021}. Expectation sensitivity is $\varphi_{g,c} = 0.35$ for all $g$ and $c$, following cross-sectional regressions of survey-measured inflation expectations on experienced inflation by income group in the ECB Consumer Expectations Survey. The results section reports robustness to alternative values.

The economic content of (\ref{eq:expect})--(\ref{eq:piexp}) is that workers who experience higher inflation, because they spend more on items whose prices are rising, form higher expectations of future inflation. When these workers subsequently reset their wages, they demand higher compensation not only for past price increases but also for expected future ones. This creates a feedback loop that amplifies the initial cost-push shock: experienced inflation raises expectations, which raise reset wages, which raise services prices, which raise the experienced inflation of the next cohort of resetting workers.

\subsection{Firms}\label{sec:firms}

Each country $c$ has two domestic production sectors, goods ($d$) and services ($s$), and an imported-essentials sector ($e$).

\noindent\textbf{Imported essentials.} The imported-essentials price is exogenous to each country and common across the union:
\begin{equation}\label{eq:pe}
    p_{e,t} = \bar{p}_e + u_t\,,
\end{equation}
where $u_t$ is the cost-push shock. This captures the common external shocks, namely energy, food, shipping, and imported intermediates, that hit all euro-area members simultaneously.

\noindent\textbf{Technology.} In each domestic sector $j \in \{d,s\}$, a representative competitive firm in country $c$ produces output $Y_{j,c,t}$ using labor from both worker types and intermediate inputs from the other sector:
\begin{equation}\label{eq:production}
    Y_{j,c,t} = A_{j,c}\!\left[\alpha_{j,c}\,L_{j,c,t}^{\frac{\rho-1}{\rho}} + (1-\alpha_{j,c})\,M_{j,c,t}^{\frac{\rho-1}{\rho}}\right]^{\frac{\rho}{\rho-1}},
\end{equation}
where $A_{j,c}$ is total factor productivity, $L_{j,c,t}$ is a composite of type-$H$ and type-$L$ labor, $M_{j,c,t}$ is intermediate input from the other domestic sector and from imported essentials, $\alpha_{j,c} \in (0,1)$ is the labor share, and $\rho > 0$ is the elasticity of substitution between labor and intermediates.

\noindent\textbf{Labor demand.} Within each sector $j$, the labor composite $L_{j,c,t}$ aggregates the two worker types with a CES aggregator of elasticity $\varepsilon_w$. Cost minimization yields the demand for type-$g$ labor: $N^d_{g,j,c,t}=\eta_{g,j,c}(W_{g,c,t}/W_{j,c,t})^{-\varepsilon_w}L_{j,c,t}$, where $W_{j,c,t}$ is the sectoral wage index. The log marginal revenue product of type-$g$ labor is
\begin{equation}\label{eq:mrpl}
    \mathrm{mrpl}_{g,c,t} = p_{j,c,t} + \log\mathrm{MP}_{L_j,c,t} - (\varepsilon_w - 1)(w_{g,c,t} - w_{j,c,t})\,,
\end{equation}
where $\mathrm{MP}_{L_j,c,t}\equiv\alpha_{j,c}Y_{j,c,t}/L_{j,c,t}$ is the physical marginal product of the labor composite and $w_{j,c,t}=\sum_g\eta_{g,j,c}w_{g,c,t}$ is the sectoral wage index. To first order around a symmetric steady state ($w_{g,c}^{ss}=w_c^{ss}$ for all $g$), the relative-wage term vanishes and $\mathrm{mrpl}_{g,c,t}\approx p_{j,c,t}+\log\mathrm{MP}_{L_j,c,t}$ for both types. The details are given in Appendix~\ref{app:firms}.

\noindent\textbf{Price setting.} Firms in each domestic sector face Calvo price adjustment. In each period, a fraction $\theta^p_{j,c}$ of firms can reset their price. The log-linearized sectoral price satisfies
\begin{equation}\label{eq:psector}
    p_{j,c,t} = \alpha_{j,c}\,w_{j,c,t} + (1-\alpha_{j,c})\sum_{j'}\xi_{j,j',c}\,p_{j',c,t}\,,
\end{equation}
where $w_{j,c,t} = \sum_g \eta_{g,j,c}\,w_{g,c,t}$ is the sectoral wage (a weighted average of type-specific wages, with weights $\eta_{g,j,c}$ equal to the employment share of type $g$ in sector $j$), and $\xi_{j,j',c}$ are input-output weights satisfying $\sum_{j'}\xi_{j,j',c} = 1$.

\noindent\textbf{Regularity condition.} I maintain throughout that services has a higher labor share than goods: $\alpha_{s,c} > \alpha_{d,c}$ for all $c$. This is an empirical regularity rather than a modeling choice: across all six euro-area countries in the sample, the services labor share exceeds 0.60 while the goods labor share is below 0.45 (Table~\ref{tab:param_country}). The condition implies that wage increases in the services sector feed into services prices more strongly than wage increases in goods feed into goods prices.

\noindent\textbf{Propagation weights.} To summarize how wage increases in each sector affect broader prices, I define the propagation weight of sector $j$ in country $c$.

\begin{definition}\label{def:propweight}
The propagation weight of sector $j$ in country $c$ is
\begin{equation}\label{eq:propweight}
    \nu_{j,c} \equiv \alpha_{j,c} \cdot e_{j,c} \cdot (1-\theta^p_{j,c})^{-1}\,,
\end{equation}
where $\alpha_{j,c}$ is the labor share, $e_{j,c}$ is the eigenvector centrality of sector $j$ in the input-output network of country $c$, which measures how much a price change in sector $j$ propagates to other sectors through intermediate-input linkages, and $(1-\theta^p_{j,c})^{-1}$ captures the amplification from nominal price rigidity.
\end{definition}

The propagation weight combines three channels through which a sectoral wage increase affects broader prices: the direct effect through the labor share ($\alpha_{j,c}$), the indirect effect through the input-output network ($e_{j,c}$), and the accumulation effect through price stickiness ($(1-\theta^p_{j,c})^{-1}$). Since services has a higher labor share ($\alpha_{s,c}>\alpha_{d,c}$), $\nu_{s,c} > \nu_{d,c}$: wage increases in services propagate more than wage increases in goods.

\subsection{Aggregation and core inflation}

Core inflation in country $c$ is the GDP-weighted average of domestic sectoral price changes:
\begin{equation}\label{eq:core}
    \pi^{core}_{c,t} \equiv \omega_{d,c}\,\Delta p_{d,c,t} + \omega_{s,c}\,\Delta p_{s,c,t}\,,
\end{equation}
where $\omega_{d,c}$ and $\omega_{s,c}$ are the GDP shares of goods and services. This corresponds to the empirical concept of HICP excluding energy and unprocessed food.

\subsection{Aggregate demand}

Log-linearizing the household's Euler equation and aggregating across types yields the IS curve for country $c$:
\begin{equation}\label{eq:IS}
    x_{c,t} = \mathbb{E}_t[x_{c,t+1}] - \sigma_c^{-1}(i_t - \mathbb{E}_t[\pi_{c,t+1}] - r^n_{c,t})\,,
\end{equation}
where $x_{c,t}$ is the output gap, $i_t$ is the common nominal interest rate, $\pi_{c,t}$ is headline inflation, $r^n_{c,t}$ is the natural real rate, and $\sigma_c$ is the aggregate intertemporal elasticity of substitution. In the simple model with complete markets, $\sigma_c=\sigma$ for all $c$; in the quantitative HANK extension, $\sigma_c$ is an effective parameter that differs across countries due to heterogeneity in the wealth distribution and marginal propensities to consume. The derivation is standard and given in Appendix~\ref{app:IS}.

\subsection{Policy}\label{sec:policy_model}

\noindent\textbf{Monetary policy.} The common central bank sets the nominal interest rate according to a Taylor rule that responds to union-wide aggregates:
\begin{equation}\label{eq:taylor}
    i_t = r^* + \phi_\pi\,\pi^{core}_{union,t} + \phi_y\,x_{union,t}\,,
\end{equation}
where
\begin{equation}\label{eq:union_agg}
    \pi^{core}_{union,t} \equiv \sum_c w_c\,\pi^{core}_{c,t}\,,\qquad x_{union,t} \equiv \sum_c w_c\,x_{c,t}\,,
\end{equation}
are GDP-weighted union averages, $r^*$ is the steady-state real rate, and $\phi_\pi > 1$, $\phi_y > 0$ are Taylor rule coefficients. The constraint that all countries face the same $i_t$ is the defining feature of the monetary union.

\noindent\textbf{Fiscal policy.} Each country $c$ has a fiscal authority that can make type-specific lump-sum transfers $T_{g,c,t}$ to its workers, financed by lump-sum taxes $\mathcal{T}_{c,t}$ on all workers, subject to a balanced-budget constraint $\sum_g \eta_{g,c}\,T_{g,c,t} = \mathcal{T}_{c,t}$. The model considers four fiscal regimes: (i)~no fiscal response ($T_{g,c,t} = 0$); (ii)~uniform transfer ($T_{H,c,t} = T_{L,c,t}$); (iii)~targeted transfer ($T_{H,c,t} = \tau_{c,t} > 0$, $T_{L,c,t} = 0$); (iv)~essentials subsidy, implemented as a reduction in the effective essentials price faced by all workers.

\subsection{Equilibrium}

An equilibrium consists of sequences of prices, quantities, and distributions for all worker types $g$, sectors $j$, countries $c$, and dates $t$, such that:
\begin{enumerate}
    \item Households maximize expected discounted utility (\ref{eq:utility}) subject to (\ref{eq:budget}), taking prices and wages as given, except when resetting their own wage.
    \item Resetting workers choose the optimal wage $W^*_{g,c,t}$ to solve (\ref{eq:wage_problem}).
    \item Non-resetting workers retain their previous nominal wage.
    \item Firms in each domestic sector choose inputs to maximize profits given technology (\ref{eq:production}) and Calvo pricing.
    \item The central bank sets $i_t$ according to (\ref{eq:taylor}).
    \item The fiscal authority in each country satisfies the balanced-budget constraint.
    \item The goods market in each country clears: domestic output equals domestic consumption plus net exports for tradable goods, or equals domestic consumption for non-tradable services.
    \item The bond market clears: $\sum_c\sum_g \eta_{g,c}\,B_{g,c,t} = \bar{B}$ where $\bar{B}$ is the exogenous aggregate bond supply.
\end{enumerate}

\subsection{Key objects}\label{sec:keyobjects}

I now define three objects that play a central role in the analysis. Each measures a different aspect of the inflation experienced by workers at the wage-reset margin.

\begin{definition}\label{def:rwei}
Reset-Weighted Experienced Inflation (RWEI) in country $c$ at date $t$ is
\begin{equation}\label{eq:rwei}
    \mathrm{RWEI}_{c,t} \equiv \sum_{g \in \{H,L\}} \omega^{reset}_{g,c} \cdot \tilde{\pi}^{exp}_{g,c,t}\,,
\end{equation}
where $\omega^{reset}_{g,c} \equiv \eta_{g,c}\,\theta_{g,c}\,/\,\bar{\theta}_c$ is the share of aggregate wage resets accounted for by type-$g$ workers, $\bar{\theta}_c \equiv \sum_g \eta_{g,c}\,\theta_{g,c}$ is the average reset frequency, and $\tilde{\pi}^{exp}_{g,c,t}$ is salient experienced inflation defined in (\ref{eq:piexp}).
\end{definition}

RWEI measures the average inflation experienced by workers who are actually in a position to renegotiate wages in the current period. It differs from average CPI inflation whenever the workers at the reset margin have different consumption baskets from the average worker. Under Assumptions~\ref{ass:baskets}--\ref{ass:reset}, type-$H$ workers account for a disproportionate share of resets (because $\theta_{H,c} > \theta_{L,c}$) and experience higher inflation from essentials price shocks (because $\alpha_{H,c,e} > \alpha_{L,c,e}$). Therefore, RWEI exceeds average inflation when essentials prices rise.

\begin{definition}\label{def:mwsi}
Marginal Wage Setter Inflation (MWSI) in country $c$ at date $t$ is
\begin{equation}\label{eq:mwsi}
    \mathrm{MWSI}_{c,t} \equiv \sum_{g \in \{H,L\}} (\omega^{reset}_{g,c} - \eta_{g,c}) \cdot \nu_{s(g),c} \cdot \tilde{\pi}^{exp}_{g,c,t}\,,
\end{equation}
where $\nu_{s(g),c}$ is the propagation weight (Definition~\ref{def:propweight}) of the sector employing type-$g$ workers.
\end{definition}

MWSI augments the reset-heterogeneity wedge with sectoral propagation: it measures how much the excess inflation experienced by resetting workers relative to the average worker propagates into core prices through the input-output network. MWSI equals $\Omega_{c,t}$ when propagation weights are homogeneous ($\nu_{s,c}=\nu_{d,c}$), and exceeds $\Omega_{c,t}$ when type-$H$ workers are concentrated in high-propagation sectors such as services. MWSI is zero in the standard model, where $\omega^{reset}_{g,c}=\eta_{g,c}$ for all $g$.

\begin{definition}\label{def:wedge}
The reset-heterogeneity wedge in country $c$ at date $t$ is
\begin{equation}\label{eq:wedge}
    \Omega_{c,t} \equiv \mathrm{RWEI}_{c,t} - \bar{\pi}_{c,t}\,,
\end{equation}
where $\bar{\pi}_{c,t} \equiv \sum_g \eta_{g,c}\,\tilde{\pi}^{exp}_{g,c,t}$ is the population-weighted average experienced inflation.
\end{definition}

The wedge measures how much more inflation is experienced by resetting workers relative to the average worker. It can equivalently be written as a weighted covariance:
\begin{equation}\label{eq:wedge_cov}
    \Omega_{c,t} = \frac{1}{\bar{\theta}_c}\,\mathrm{Cov}_{\eta_c}\!\left(\theta_{g,c},\, \tilde{\pi}^{exp}_{g,c,t}\right),
\end{equation}
where the covariance is taken across worker types weighted by population shares $\eta_{g,c}$. This covariance form makes explicit that the wedge depends on the correlation between how frequently a worker resets wages and how much inflation that worker experiences, not on either margin alone. When this correlation is positive, which is the empirically dominant case, the wedge is positive and the standard model, which implicitly sets $\Omega_{c,t} = 0$, underestimates wage pressure.

\section{Theoretical Results}\label{sec:theory}

This section states the main theoretical results. All proofs are in the Appendix.

\subsection{The heterogeneous wage Phillips curve}

The first step is to derive the wage Phillips curve for each type separately. Quasi-differencing the reset-wage equation (\ref{eq:wreset}) and substituting into the wage-inflation identity (\ref{eq:piw_type}) yields a type-specific wage Phillips curve (derived in Appendix~\ref{app:proofs}):
\begin{equation}\label{eq:twpc}
    \pi^w_{g,c,t} = \beta\,\mathbb{E}_t[\pi^w_{g,c,t+1}] + \kappa_{g,c}\,\hat{\omega}_{g,c,t} + \theta_{g,c}\,\pi^{exp}_{g,c,t}\,.
\end{equation}
There are three terms. The first is standard: expected future wage inflation. The second involves the type-specific real marginal cost gap, $\hat{\omega}_{g,c,t} \equiv \mathrm{mrpl}_{g,c,t} - w_{g,c,t} + p_{g,c,t}$, with slope $\kappa_{g,c} \equiv \theta_{g,c}(1-\beta\Theta_{g,c})/\Theta_{g,c}$. The third term, $\theta_{g,c}\,\pi^{exp}_{g,c,t}$, is the inflation experienced by type $g$, scaled by the reset probability. It is present because resetting workers target their own price index: when their cost of living rises, they demand higher nominal wages even if the aggregate output gap is unchanged.

Aggregating (\ref{eq:twpc}) across types using the population weights $\eta_{g,c}$ yields the paper's first result.

\begin{proposition}\label{prop:hwpc}
\textbf{(Heterogeneous Wage Phillips Curve.)} Aggregate wage inflation in country $c$ satisfies
\begin{equation}\label{eq:hwpc}
    \pi^w_{c,t} = \beta\,\mathbb{E}_t[\pi^w_{c,t+1}] + \tilde{\kappa}_c\,x_{c,t} + \bar{\theta}_c\,\bar{\pi}_{c,t} + \bar{\theta}_c\,\Omega_{c,t}\,,
\end{equation}
where $\tilde{\kappa}_c$ is the aggregate Phillips curve slope, defined as
\begin{equation}\label{eq:kappa_agg}
    \tilde{\kappa}_c \equiv \sum_g \eta_{g,c}\,\kappa_{g,c} = \sum_g \eta_{g,c}\,\frac{\theta_{g,c}(1-\beta\Theta_{g,c})}{\Theta_{g,c}}\,,
\end{equation}
and $\Omega_{c,t}$ is the reset-heterogeneity wedge (Definition~\ref{def:wedge}). The fourth term, $\bar{\theta}_c\,\Omega_{c,t}$, is absent from the standard New Keynesian wage Phillips curve. It is identically zero if and only if at least one of the following conditions holds:
\begin{enumerate}
    \item All worker types have identical consumption baskets: $\alpha_{H,c,i} = \alpha_{L,c,i}$ for all $i$.
    \item All worker types have identical reset frequencies: $\theta_{H,c} = \theta_{L,c}$.
    \item The price shock is uniform across all consumption categories: $\Delta p_{i,c,t}$ is the same for all $i$.
\end{enumerate}
Under Assumptions~\ref{ass:baskets}--\ref{ass:reset} and a positive imported-essentials price shock ($u_t > 0$), $\Omega_{c,t} > 0$ and the standard model understates wage inflation.
\end{proposition}

The proof is in Appendix~\ref{app:proofs}. The key step is the aggregation of the experienced-inflation terms across types, which produces the RWEI object rather than average inflation (see equation (\ref{eq:wedge_cov})). The result is stated for two types for expositional clarity, but extends immediately to any number of types $G\geq 2$: the wedge is $(1/\bar\theta_c)\mathrm{Cov}_{\eta_c}(\theta_{g,c},\tilde\pi^{exp}_{g,c,t})$, which is well-defined for any discrete distribution over worker types. The quantitative model uses $G=5$.

A crucial feature of Proposition~\ref{prop:hwpc} is that the aggregation error is shock-dependent. The wedge $\Omega_{c,t}$ is large when the cost-push shock is concentrated on necessities and zero when the shock raises all prices uniformly. The standard wage Phillips curve therefore works well for demand shocks, which tend to raise prices broadly, but fails for supply shocks that hit essentials. Since the 2021--2023 episode was driven by energy and food, the standard model missed it. This shock-dependence is a qualitative prediction that no recalibration of the standard model can replicate: the representative-agent Phillips curve has no mechanism for producing different forecast errors depending on the composition of the shock.

Figure~\ref{fig:mechanism} illustrates the transmission mechanism. The standard model collapses the dashed box into a single representative agent; this paper opens it.

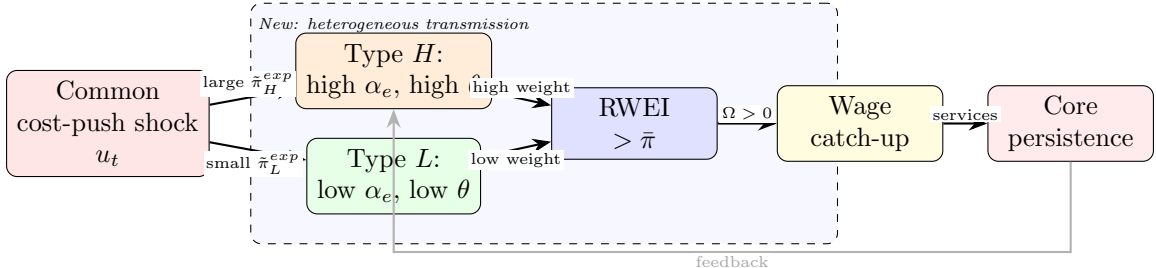
\begin{figure}[htbp]
\centering
\begin{tikzpicture}[
    block/.style={draw, rounded corners, minimum width=2.2cm, minimum height=0.8cm, align=center, font=\footnotesize},
    arrow/.style={-{Stealth[length=2.5mm]}, thick},
    label/.style={font=\tiny, midway, fill=white, inner sep=1pt}
]
\node[block, fill=red!10] (shock) at (0,0) {Common\\cost-push shock\\$u_t$};

\node[draw, dashed, rounded corners, minimum width=7.8cm, minimum height=3.2cm, fill=blue!3] (box) at (5.8,0) {};
\node[font=\tiny\itshape, anchor=north west] at (1.85,1.55) {New: heterogeneous transmission};

\node[block, fill=orange!15] (piH) at (3.8,0.7) {Type $H$:\\high $\alpha_e$, high $\theta$};
\node[block, fill=green!10] (piL) at (3.8,-0.7) {Type $L$:\\low $\alpha_e$, low $\theta$};

\node[block, fill=blue!12] (rwei) at (7.0,0) {RWEI\\$>\bar\pi$};

\node[block, fill=yellow!15] (wage) at (10.0,0) {Wage\\catch-up};

\node[block, fill=red!8] (core) at (12.8,0) {Core\\persistence};

\draw[arrow] (shock) -- (piH) node[label, above, pos=0.45] {large $\tilde\pi^{exp}_H$};
\draw[arrow] (shock) -- (piL) node[label, below, pos=0.45] {small $\tilde\pi^{exp}_L$};
\draw[arrow] (piH) -- (rwei) node[label, above] {high weight};
\draw[arrow] (piL) -- (rwei) node[label, below] {low weight};
\draw[arrow] (rwei) -- (wage) node[label, above] {$\Omega>0$};
\draw[arrow] (wage) -- (core) node[label, above] {services};

\draw[arrow, gray!60] (core.south) -- ++(0,-1.2) -| (piH.south) node[label, below, pos=0.25, text=gray!70] {feedback};
\end{tikzpicture}
\caption{Transmission mechanism.}\label{fig:mechanism}
\begin{flushleft}
\footnotesize Note: A common cost-push shock hits both worker types. Type $H$, with higher essentials expenditure share ($\alpha_{H,e}>\alpha_{L,e}$) and higher reset frequency ($\theta_H>\theta_L$), experiences more inflation and accounts for a disproportionate share of wage resets. RWEI exceeds average CPI because reset weights overweight type $H$. The standard model collapses the dashed box into a representative agent, setting $\Omega=0$. The feedback loop from core inflation to type-$H$ experienced inflation sustains the wedge beyond the initial shock.
\end{flushleft}
\end{figure}

Equation (\ref{eq:hwpc}) nests the standard wage Phillips curve as a special case. The first three terms on the right, expected future wage inflation, the output gap, and average experienced inflation, are present in the standard model. The fourth term is new. It says that aggregate wage inflation depends not only on how much inflation there is on average, but on which workers experience the inflation. When the cost-push shock concentrates on items consumed disproportionately by frequently resetting workers, the covariance in (\ref{eq:wedge_cov}) is positive, and the standard model underestimates wage pressure.

For the two-type case, the wedge has the explicit closed form
\begin{equation}\label{eq:wedge_cf}
    \Omega_{c,t} = \frac{\eta_{H,c}\eta_{L,c}}{\bar{\theta}_c}\,(\theta_{H,c}-\theta_{L,c})\,(\alpha_{H,c,e}-\alpha_{L,c,e})\,(\lambda_e\Delta p_{e,c,t} - \bar{\lambda}\Delta\bar{p}_{c,t})\,.
\end{equation}

Figure~\ref{fig:wedge_numerical} provides a concrete illustration. In this example, the two types have equal population shares but different baskets and reset frequencies. Because type $H$ accounts for 81 percent of wage resets (due to higher $\theta_H$) while experiencing 15 percent inflation (due to higher $\alpha_{H,e}$), RWEI exceeds average CPI by 2.5 percentage points. The standard model, which weights all workers equally, misses this.

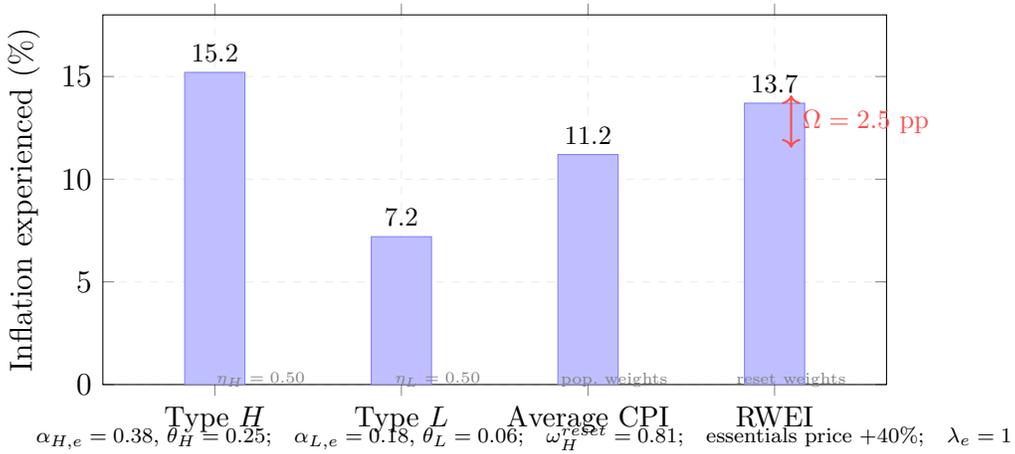
\begin{figure}[htbp]
\centering
\begin{tikzpicture}
\begin{axis}[
    ybar,
    width=12cm, height=6.5cm,
    bar width=0.8cm,
    ylabel={Inflation experienced (\%)},
    symbolic x coords={Type $H$, Type $L$, Average CPI, RWEI},
    xtick=data,
    ymin=0, ymax=18,
    nodes near coords,
    nodes near coords style={font=\footnotesize},
    every node near coord/.append style={anchor=south},
    legend style={at={(0.02,0.98)}, anchor=north west, font=\footnotesize, draw=none},
    grid=major, grid style={dashed, gray!15},
    x tick label style={font=\small},
    enlarge x limits=0.2,
]
\addplot[fill=blue!25, draw=blue!50] coordinates {(Type $H$,15.2) (Type $L$,7.2) (Average CPI,11.2) (RWEI,13.7)};
\end{axis}
\draw[<->, thick, red!70] (9.15,3.15) -- (9.15,3.85) node[midway, right, font=\footnotesize, red!70] {$\Omega=2.5$ pp};
\node[font=\tiny, anchor=north, text=gray] at (2.1,0.3) {$\eta_H=0.50$};
\node[font=\tiny, anchor=north, text=gray] at (4.45,0.3) {$\eta_L=0.50$};
\node[font=\tiny, anchor=north, text=gray] at (6.8,0.3) {pop.\ weights};
\node[font=\tiny, anchor=north, text=gray] at (9.15,0.3) {reset weights};
\node[font=\scriptsize, anchor=north] at (5.6,-0.4) {$\alpha_{H,e}=0.38$, $\theta_H=0.25$;\quad $\alpha_{L,e}=0.18$, $\theta_L=0.06$;\quad $\omega^{reset}_H=0.81$;\quad essentials price $+40\%$;\quad $\lambda_e=1$};
\end{tikzpicture}
\caption{Numerical illustration of the wedge.}\label{fig:wedge_numerical}
\begin{flushleft}
\footnotesize Note: Both types face the same 40\% essentials price shock. Parameters: $\alpha_{H,e}=0.38$, $\alpha_{L,e}=0.18$, $\theta_H=0.25$, $\theta_L=0.06$, $\eta_H=\eta_L=0.5$, salience $\lambda_e=1$. Type $H$ experiences higher inflation and resets more often. RWEI $=\omega^{reset}_H\cdot\pi^{exp}_H+\omega^{reset}_L\cdot\pi^{exp}_L=13.7\%$ versus average CPI $=11.2\%$. The wedge $\Omega=2.5$ percentage points measures the additional wage pressure that the standard model, which weights all workers by population share, misses.
\end{flushleft}
\end{figure}

The wedge is the product of three measurable objects: the gap in reset frequencies ($\theta_{H,c}-\theta_{L,c}$, determined by labor market institutions), the gap in essentials expenditure shares ($\alpha_{H,c,e}-\alpha_{L,c,e}$, determined by consumption patterns), and the relative price change of essentials. Specifically, the last factor equals $\sum_i\lambda_i(\alpha_{H,c,i}-\alpha_{L,c,i})\Delta p_{i,c,t}$, which is positive when essentials prices rise faster than other prices. All three gaps must be nonzero for the wedge to be active. When expectations respond to experienced inflation ($\varphi_{g,c}>0$), the wedge is amplified by a factor $(1+\bar{\varphi}_c)$, as shown in Appendix~\ref{app:amplifier}.

The static wedge generates dynamic effects through two channels. First, the wedge-driven wage increase at date $t$ raises services prices, which feeds back into type-$H$'s experienced inflation at date $t+1$, sustaining the wedge beyond the original shock. Second, the Calvo structure implies that only a fraction $\theta_{H,c}$ of type-$H$ workers reset at date $t$; the rest carry pre-shock wages and will demand catch-up at their next reset. The cumulative wage response to a one-time shock $u_0$ is
\begin{equation}\label{eq:cumwage}
    \sum_{h=0}^{\infty}\beta^h\pi^w_{c,h} = \frac{\tilde{\kappa}_c}{1-\beta}\,\bar{x} + \frac{\bar{\theta}_c}{1-\beta}\,\bar{\pi}_c(u_0) + \frac{\bar{\theta}_c}{1-\beta\rho_\Omega}\,\Omega_{c,0}\,,
\end{equation}
where $\rho_\Omega \in (0,1)$ is the persistence of the wedge. Equation (\ref{eq:cumwage}) is the formal version of the decomposition previewed in equation (\ref{eq:decomp_channels}): the first term is demand, the second is the level catch-up channel of \citet{afrouzi2024}, and the third is the composition wedge introduced here. The two inflation channels are orthogonal: the level channel depends on average inflation $\bar{\pi}_c$, the composition channel on the wedge $\Omega_c$. In the linearized model, $\rho_\Omega$ is close to the persistence of the shock itself; the wedge's quantitative contribution comes primarily through its amplitude rather than through differential persistence. In a nonlinear model with convex catch-up, $\rho_\Omega$ can exceed $\rho_{shock}$ because workers who have fallen behind bargain harder at each subsequent reset, making the wedge self-reinforcing.

\subsection{Cross-country dispersion in inflation persistence}

Because the wedge $\Omega_{c,t}$ depends on country-specific labor market composition, which differs across the monetary union, a common imported shock generates different inflation persistence in different countries. To state this precisely, denote by $R_c$ the standard-model pass-through coefficient, defined as the elasticity of cumulative core inflation to the shock in the representative-agent benchmark, and by $S_c > 0$ a structural amplification factor that depends on the input-output structure and price rigidity parameters of country $c$.

\begin{proposition}\label{prop:disp}
\textbf{(Cross-Country Dispersion.)} Consider $N$ countries hit by the same shock $u_t$. The cross-country variance of cumulative core inflation decomposes as
\begin{equation}\label{eq:vardecomp}
    \mathrm{Var}_c(\Pi^{core}_c) = \underbrace{\mathrm{Var}_c(R_c\,u)}_{\text{standard}} + \underbrace{\mathrm{Var}_c(S_c\bar{\theta}_c\Omega_c\,u)}_{\text{wedge}} + 2\,\mathrm{Cov}_c(R_c,\,S_c\bar{\theta}_c\Omega_c)\,u^2\,,
\end{equation}
where $R_c$ is the standard-model pass-through and $S_c > 0$ is a structural amplification factor. The standard model ($\Omega_c = 0$) underestimates cross-country dispersion whenever $\mathrm{Var}_c(S_c\bar\theta_c\Omega_c) + 2\,\mathrm{Cov}_c(R_c,\,S_c\bar\theta_c\Omega_c) > 0$, a condition that holds when the covariance between the standard pass-through and the wedge is non-negative.
\end{proposition}

The condition $\mathrm{Cov}_c(R_c,S_c\bar\theta_c\Omega_c)\geq 0$ holds in the calibration: countries with higher trade openness tend to have both higher standard pass-through $R_c$ and higher wedge $\Omega_c$, because imported essentials are the common driver of both channels. In the six-country sample, the empirical correlation between $R_c$ and $\Omega_c$ is 0.38.

The proof is in Appendix~\ref{app:disp}. The variance decomposition~(\ref{eq:vardecomp}) says that the standard OCA framework misses a source of cross-country asymmetry: variation in the within-country composition of wage resetting and consumption. Two countries can have the same average openness, the same wage-setting institutions, and the same Phillips curve slopes, yet exhibit different core inflation persistence if their RWEI differs. This source of asymmetry is measurable from micro data and is distinct from all classical OCA criteria.

\begin{corollary}\label{cor:agg}
No representative-agent economy, regardless of the choice of $\bar{\theta}$ and $\bar{\alpha}$, replicates the aggregate wage dynamics of the heterogeneous-reset economy for all shock compositions. The best-fitting representative agent systematically underestimates persistence for necessity shocks and overestimates it for non-necessity shocks.
\end{corollary}

The proof is in Appendix~\ref{app:corollary}.

Corollary~\ref{cor:agg} connects two macroeconomic puzzles that the literature has treated separately. During 2021--2023, the standard model underpredicted wage growth and core inflation persistence after a positive essentials price shock: $\Omega > 0$. During 2014--2016, when oil prices fell by more than 60 percent, the same model overpredicted wage disinflation: bottom-quintile workers experienced more deflation than the average worker, so $\Omega < 0$, and the standard model was too pessimistic about the disinflationary effect. The 2014--2016 episode is the ``missing disinflation'' puzzle that motivated a large literature on anchored expectations and downward wage rigidity. The present framework offers a complementary explanation that requires neither expectation anchoring nor asymmetric rigidity: the standard model simply uses the wrong price index for the workers at the reset margin. One sufficient statistic, RWEI, accounts for the sign of the error in both episodes.

\subsection{Sufficient statistic for inflation persistence}

Tracing the wedge through the sectoral price system (\ref{eq:psector}) and using the propagation weights from Definition~\ref{def:propweight}, the additional core inflation persistence generated by the wedge can be expressed in terms of a single sufficient statistic.

\begin{proposition}\label{prop:suff}
\textbf{(Sufficient Statistic.)} To first order around the steady state, the additional cumulative core inflation in country $c$ relative to the standard model is
\begin{equation}\label{eq:suff}
    \Pi^{core}_c(u) - \Pi^{core,RA}_c(u) = \psi\,\mathrm{MWSI}_c(u) + O(\|u\|^2)\,,
\end{equation}
where $\psi > 0$ is a structural coefficient, common across countries to first order, and $\mathrm{MWSI}_c(u)$ is Marginal Wage Setter Inflation (Definition~\ref{def:mwsi}).
\end{proposition}

The proof is in Appendix~\ref{app:suff}. The result says that a researcher with data on expenditure shares $\{\alpha_{g,c,i}\}$, reset frequencies $\{\theta_{g,c}\}$, and propagation weights $\{\nu_{s,c}\}$ can compute MWSI for any country and any shock composition, and obtain a first-order prediction of how much additional cumulative core inflation that country will exhibit relative to the standard model, without solving any general equilibrium model. This portability is the practical value of the sufficient statistic. In the linearized model, the wedge raises the amplitude of the inflation response at each horizon; the half-life is governed by the Calvo price-adjustment parameters, which are common to both models. Nonlinear catch-up dynamics can generate additional half-life differences by making the wedge more persistent than the underlying shock, but this channel operates at second order in $\|u\|$ and requires solving the model nonlinearly.

Proposition~\ref{prop:suff} has a direct implication for inflation forecasting. A central bank that uses the standard wage Phillips curve to forecast core inflation after an imported cost-push shock will make a systematic forecasting error equal to $\psi\,\mathrm{MWSI}_c(u)$. The sign of the error is predicted by Corollary~\ref{cor:agg}: the standard model under-predicts after necessity shocks and over-predicts after non-necessity shocks. The magnitude of the error is increasing in the shock size: since $\mathrm{MWSI}_c(u)$ is proportional to $u$, the forecasting error is proportional to $u^2$ as a fraction of total core inflation. The wedge is therefore a tail-risk amplifier: it is negligible for small shocks but first-order for the large supply shocks that matter most for policy. This state-dependence parallels the role of MPC heterogeneity in the HANK literature, where marginal propensities to consume matter little for small interest rate changes but dominate the transmission of large fiscal transfers \citep{kaplan2018}.

\subsection{Implications for monetary and fiscal policy}\label{sec:policy_theory}

The wedge has three implications for optimal stabilization policy, each of which follows directly from the structure of equation (\ref{eq:hwpc}). The formal derivations are in Appendix~\ref{app:policy} and~\ref{app:union}.

\medskip\noindent\textbf{The interest rate alone is insufficient.} Setting $x_{c,t}=0$ in (\ref{eq:hwpc}), wage inflation remains $\bar{\theta}_c\bar{\pi}_{c,t}+\bar{\theta}_c\Omega_{c,t}\neq 0$ whenever $\Omega_{c,t}>0$. The wedge is a cross-sectional object that depends on which workers experience which inflation, so the aggregate interest rate cannot directly close it. Divine coincidence fails: closing the output gap no longer suffices to stabilize inflation when the wedge is positive.

\begin{proposition}\label{prop:policy}
\textbf{(Optimal Policy Mix.)} Consider a policymaker with access to the interest rate $i_t$ and a targeted transfer $\tau_{H,c,t}$ to high-RWEI workers. Then:
\begin{enumerate}
    \item[(i)] The optimal Taylor coefficient is decreasing in RWEI: $\partial\phi^*_\pi/\partial\,\mathrm{RWEI}_c<0$.
    \item[(ii)] The targeted transfer is net disinflationary at the medium horizon whenever
\begin{equation}\label{eq:disfl}
    \underbrace{\psi\,\omega^{reset}_{H,c}\sum_{h=0}^{H^*}\beta^h}_{\text{cumulative wedge reduction}} > \underbrace{\tilde{\kappa}_c\,\mathrm{MPC}_{H,c}\,\eta_{H,c}\,C_{H,c}/\bar{C}_c}_{\text{one-period demand effect}}\,.
\end{equation}
    \item[(iii)] The optimal policy mix uses both instruments whenever $\mathrm{RWEI}_c$ exceeds a threshold $\mathrm{RWEI}^*$ that equates the marginal benefit and cost of the transfer.
    \item[(iv)] No interest rate alone achieves the first best: $\inf_{\phi_\pi}\mathrm{Var}(\pi^{core}_c) > 0$ whenever $\Omega_{c,t}>0$.
\end{enumerate}
\end{proposition}

The proof is in Appendix~\ref{app:policy}. Part~(i) says that when RWEI is high, the output cost of achieving a given inflation reduction rises, so the optimal monetary response is less aggressive. Part~(ii) says that the transfer works because the wedge-reduction benefit accumulates over many quarters while the demand cost is a one-time flow. For euro-area calibrations with $H^*\geq 6$ quarters, the cumulative benefit substantially exceeds the one-period cost, consistent with the quantitative finding that targeted transfers reduce welfare loss by 28--32 percent (Table~\ref{tab:policy}).

\medskip\noindent\textbf{A new OCA criterion.} In the monetary union, the optimal country-specific interest rate $i^*_{c,t}$ is decreasing in $\mathrm{RWEI}_c$. Since RWEI differs across members, countries prefer different rates. The common rate over-tightens for high-RWEI countries, specifically Spain and the Netherlands, and under-tightens for low-RWEI countries such as France. The welfare cost of common policy decomposes as
\begin{equation}\label{eq:oca}
    \mathcal{L}_{union}(i^*_t)-\sum_c w_c\,\mathcal{L}_c(i^*_{c,t})=\underbrace{\Gamma_1\,\mathrm{Var}_c(R_c)\,u^2}_{\text{standard OCA cost}}+\underbrace{\Gamma_2\,\mathrm{Var}_c(\bar{\theta}_c\Omega_c)\,u^2}_{\text{RWEI-dispersion cost}}\,,
\end{equation}
where $\Gamma_1,\Gamma_2>0$ are structural coefficients. The second term is the cost of RWEI dispersion, a new OCA criterion that is measurable from micro data, distinct from all classical criteria, and captures within-country composition rather than between-country asymmetry. Unlike the between-country friction heterogeneity studied by \citet{kekre2022}, RWEI dispersion arises from the interaction of common shocks with within-country labor market composition, and it is active even when all countries have identical aggregate labor market frictions. Country-level targeted transfers scaled by $\mathrm{RWEI}_c-\overline{\mathrm{RWEI}}$ eliminate this cost. Even when the external shock is symmetric, the optimal fiscal response is asymmetric because within-country labor market composition differs across members.

\medskip\noindent\textbf{Welfare cost of ignoring the wedge.} The welfare loss from using the standard Phillips curve instead of the heterogeneous one can be expressed as a function of the shock size. A second-order approximation of the representative household's utility around the efficient allocation yields the standard loss function $\mathcal{L}_c=\sum_t\beta^t[(\varepsilon_p/\kappa_p)\pi^{core\,2}_{c,t}+(\sigma+\varphi_n)x^2_{c,t}]$ following \citet{woodford2003}. The wedge enters $\pi^{core}_{c,t}$ through the Phillips curve, so the additional loss from ignoring it is
\begin{equation}\label{eq:welfare_wedge}
\Delta\mathcal{L}_c \propto \sum_t\beta^t\bigl[\psi\,\mathrm{MWSI}_c(u_t)\bigr]^2 \propto \psi^2\,\mathrm{MWSI}_c^2\cdot\frac{1}{1-\beta\rho_\Omega^2}\,.
\end{equation}
Two properties follow. First, the welfare cost is proportional to $\mathrm{MWSI}^2$ from equation~(\ref{eq:mwsi}), which is convex in the shock size $u$: doubling the shock quadruples the welfare cost of ignoring the wedge. Second, across countries in the union, the aggregate welfare cost from using the standard Phillips curve is proportional to $\mathbb{E}_c[\mathrm{MWSI}_c^2]=(\mathbb{E}_c[\mathrm{MWSI}_c])^2+\mathrm{Var}_c(\mathrm{MWSI}_c)$. Cross-country dispersion in MWSI therefore raises the welfare cost even if the mean MWSI is moderate. These properties formalize the tail-risk amplifier interpretation: the wedge matters most when it matters most.

\begin{proposition}\label{prop:twoinstopt}
\textbf{(Optimal Two-Instrument Policy.)} Suppose the policymaker has access to the interest rate $i_t$ and a per-period essentials subsidy $\tau_t$ financed by lump-sum taxation. Then:
\begin{enumerate}
    \item[(i)] The first-best allocation, $\pi^{core}_{c,t}=x_{c,t}=0$, cannot be achieved with the interest rate alone whenever $\Omega_{c,t}>0$.
    \item[(ii)] The first-best can be achieved with both instruments. The optimal subsidy is
\begin{equation}\label{eq:opt_subsidy}
    \tau^*_{c,t} = \frac{\Omega_{c,t}}{\omega^{reset}_{H,c}\,(\alpha_{H,c,e}-\alpha_{L,c,e})\,\lambda_e}\,,
\end{equation}
which is computable from the same micro data that constructs RWEI, without solving the general equilibrium model.
    \item[(iii)] The welfare gain from adding the subsidy instrument is proportional to $\mathrm{MWSI}_c^2$, convex in the shock size.
\end{enumerate}
\end{proposition}

The proof is in Appendix~\ref{app:twoinstopt}. Part~(i) restates Proposition~\ref{prop:policy}(iv). Part~(ii) is the key new result: the subsidy that closes the wedge exactly is the one that equalizes experienced inflation across worker types at the reset margin. Since $\Omega_{c,t}$ depends on the gap $\alpha_{H,c,e}-\alpha_{L,c,e}$ in essentials expenditure, the subsidy that closes the gap is proportional to $\Omega$ itself, scaled by the expenditure gradient and the reset share. All components are observable from the same three data sources that construct RWEI: household expenditure surveys, wage-setting calendars, and input-output tables. Part~(iii) follows from equation~(\ref{eq:welfare_wedge}): the welfare loss from the wedge is proportional to $\mathrm{MWSI}^2$, so eliminating the wedge recovers that loss exactly.

\medskip\noindent\textbf{CPI indexation does not close the wedge.} Several euro-area countries use automatic wage indexation to aggregate CPI. Under Calvo wage setting with indexation parameter $\gamma\in[0,1]$, non-resetting workers of type $g$ receive the automatic adjustment $w_{g,c,t}=w_{g,c,t-1}+\gamma\,\bar\pi_{c,t-1}$, where $\bar\pi_{c,t-1}$ is lagged aggregate CPI inflation. The heterogeneous wage Phillips curve~(\ref{eq:hwpc}) becomes
\begin{equation}\label{eq:hwpc_index}
    \pi^w_{c,t}-\gamma\,\bar\pi_{c,t-1} = \beta\,\mathbb{E}_t[\pi^w_{c,t+1}-\gamma\,\bar\pi_{c,t}] + \tilde\kappa_c\,x_{c,t} + \bar\theta_c\,(1-\gamma)\,\bar\pi_{c,t} + \bar\theta_c\,\Omega_{c,t}\,.
\end{equation}
The third term, the level channel, is scaled by $(1-\gamma)$ and vanishes under full indexation ($\gamma=1$). The fourth term, the composition channel, is unaffected. CPI indexation compensates all workers for average inflation, but it does not compensate type-$H$ workers for the gap between their experienced inflation and the average. Under type-specific indexation, where each type is indexed to its own price index $\pi^{exp}_{g,c,t}$, both the level and the composition channels vanish. The distinction is sharp: CPI indexation eliminates the representative-agent component of wage catch-up while leaving the heterogeneous-agent component, the wedge, fully intact.

\section{Quantitative Model, Calibration, and Computation}\label{sec:quant}

\subsection{Model}

The quantitative model extends the environment of Section~\ref{sec:model} in three respects: households face idiosyncratic productivity risk and solve a consumption-savings problem; the reset wage includes a nonlinear catch-up term; and goods are tradable within the union. The theory in Sections~\ref{sec:model}--\ref{sec:theory} uses $G=2$ worker types for analytical tractability; the quantitative model uses $G=5$ income quintiles to capture the full expenditure gradient observed in the Eurostat Household Budget Survey.

\subsubsection{Household problem}

A quintile-$g$ household in country $c$ with beginning-of-period assets $a$ and idiosyncratic productivity $e$ solves
\begin{equation}\label{eq:bellman}
    V_{g,c,t}(a,e) = \max_{c_{g},\,a'}\left\{\frac{c_g^{1-\sigma}}{1-\sigma}-\chi\frac{n_g^{1+\varphi_n}}{1+\varphi_n}+\beta\,\mathbb{E}_t\bigl[V_{g,c,t+1}(a',e')\bigr]\right\}
\end{equation}
subject to
\begin{equation}\label{eq:budget_quant}
    P_{g,c,t}\,c_g + a' = (1+r_{c,t})\,a + W_{g,c,t}\,e\,n_g + T_{g,c,t} - \mathcal{T}_{c,t}\,,\qquad a'\ge 0\,,
\end{equation}
where $P_{g,c,t}$ is the quintile-specific price index, $r_{c,t}$ is the real return on the liquid asset, $W_{g,c,t}$ is the nominal wage, $T_{g,c,t}$ is the targeted transfer, and $\mathcal{T}_{c,t}$ is a lump-sum tax. Idiosyncratic productivity follows $\log e'=\rho_e\log e+\sigma_e\varepsilon'$, $\varepsilon'\sim\mathcal{N}(0,1)$, discretized on 7 states using the Rouwenhorst method with $\rho_e=0.966$ and $\sigma_e=0.5$. Quintile $g$ consumes a Cobb-Douglas basket over essentials, goods, and services with shares $\alpha_{g,c,e},\alpha_{g,c,d},\alpha_{g,c,s}$, so the price index is (\ref{eq:pindex}). The quantitative budget constraint (\ref{eq:budget_quant}) extends the simple-model budget constraint (\ref{eq:budget}) by adding idiosyncratic productivity, the borrowing constraint, and a real rather than nominal return.

\subsubsection{Wage block}

The reset wage in the quantitative model includes a nonlinear catch-up term:
\begin{equation}\label{eq:nlreset}
    w^*_{g,c,t} = w^{*,lin}_{g,c,t} + b'\cdot\max\{p_{g,c,t}-w_{g,c,t-1},\;0\}\,,
\end{equation}
where $w^{*,lin}_{g,c,t}$ is the linear reset wage from (\ref{eq:wreset}) and $b'=0.15$. This captures the empirical regularity that negotiated wage increases are convex in the real-wage gap: workers who have fallen further behind demand proportionally larger raises.

\subsubsection{Trade}

Goods produced in country $c$ are tradable within the union. Country $c$'s demand for country $c'$'s goods takes the CES form with trade elasticity $\varepsilon_{trade}=1.5$. Services are non-tradable.

\subsubsection{Equilibrium}

An equilibrium consists of sequences $\{i_t, \{p_{j,c,t}, w_{g,c,t}, x_{c,t}, D_{g,c,t}\}_{j,g,c}\}_{t=0}^\infty$ such that: (i)~households solve (\ref{eq:bellman})--(\ref{eq:budget_quant}); (ii)~wages satisfy the Calvo structure with type-specific resets; (iii)~sectoral prices satisfy (\ref{eq:psector}); (iv)~the common interest rate follows (\ref{eq:taylor}); (v)~goods and asset markets clear in each country; (vi)~trade is balanced within the union. The distribution $D_{g,c,t}(a,e)$ of type-$g$ households over assets and productivity in country $c$ evolves according to the policy functions from (\ref{eq:bellman}).

The model nests the standard multi-country NK model when $\alpha_{g,c,i}=\bar\alpha_{c,i}$ for all $g,c,i$; $\theta_{g,c}=\bar\theta_{c}$ for all $g,c$; and $\varphi_{g,c}=b'=0$.

\subsection{Calibration}

Tables~\ref{tab:param_common} and~\ref{tab:param_country} report the calibrated parameters.

\begin{table}[htbp]
\centering\small
\caption{Common parameters}\label{tab:param_common}
\begin{tabular}{llcc}
\toprule
Symbol & Description & Value & Source \\
\midrule
$\beta$ & Discount factor & 0.99 & Standard \\
$\sigma$ & Relative risk aversion & 2.0 & Standard \\
$\varphi_n$ & Inverse Frisch elasticity & 0.5 & Chetty et al.\ (2011) \\
$\rho_e$ & Persistence, idiosyncratic prod.\ & 0.966 & Flod\'{e}n--Lind\'{e} (2001) \\
$\sigma_e$ & Std.\ dev., idiosyncratic prod.\ & 0.50 & Flod\'{e}n--Lind\'{e} (2001) \\
$\phi_\pi$ & Taylor rule, inflation & 1.50 & Standard \\
$\phi_y$ & Taylor rule, output gap & 0.125 & Standard \\
$\varepsilon_{trade}$ & Trade elasticity & 1.50 & Feenstra et al.\ (2018) \\
$b'$ & Nonlinear catch-up & 0.15 & Buchheim et al.\ (2022) \\
$\bar\varphi$ & Expectation sensitivity & 0.35 & ECB CES \\
$\lambda_e$ & Essentials salience & 1.30 & D'Acunto et al.\ (2021) \\
\bottomrule
\end{tabular}

\medskip
\begin{flushleft}
\footnotesize Note: Parameters common to all six euro-area countries. $\sigma$: inverse elasticity of intertemporal substitution. $\varphi_n$: inverse Frisch elasticity of labor supply. Shock persistence and standard deviation matched to the 2021--2023 euro-area energy price episode. Quarterly frequency.
\end{flushleft}

\end{table}

\begin{table}[htbp]
\centering\small
\caption{Country-specific parameters}\label{tab:param_country}
\begin{tabular}{lcccccc}
\toprule
& DE & FR & IT & ES & NL & BE \\
\midrule
\multicolumn{7}{l}{\text{Essentials share $\alpha_{g,e}$ by income quintile}} \\
\quad Q1 (bottom) & .40 & .38 & .42 & .43 & .37 & .39 \\
\quad Q3 (middle) & .28 & .27 & .30 & .30 & .26 & .28 \\
\quad Q5 (top) & .17 & .16 & .18 & .16 & .15 & .17 \\[3pt]
\multicolumn{7}{l}{\text{Reset probability $\theta_g$, quarterly, by quintile}} \\
\quad Q1 & 0.28 & 0.24 & 0.20 & 0.32 & 0.30 & 0.26 \\
\quad Q3 & 0.16 & 0.14 & 0.11 & 0.18 & 0.17 & 0.15 \\
\quad Q5 & 0.07 & 0.06 & 0.05 & 0.07 & 0.07 & 0.06 \\[3pt]
$\alpha_{s,c}$, labor share & 0.63 & 0.60 & 0.62 & 0.65 & 0.66 & 0.61 \\
$w_c$, GDP weight & 0.29 & 0.20 & 0.15 & 0.10 & 0.07 & 0.04 \\
\bottomrule
\end{tabular}

\medskip
\begin{flushleft}
\footnotesize Note: Expenditure shares: Eurostat HBS 2020, by income quintile. Reset probabilities: ECB wage tracker, OECD/AIAS ICTWSS, converted to quarterly probabilities from average contract duration. Quintiles Q2 and Q4 interpolated linearly. Sectoral labor shares: FIGARO input-output tables. GDP weights: Eurostat, 2019.
\end{flushleft}
\end{table}

The shock is calibrated to match the observed essentials HICP path over 2021:1--2023:4, modeled as an AR(2) with peak of 25 percent. The model does not target any core inflation or wage outcome from the episode.

\subsection{Computation}

I solve the model using the sequence-space Jacobian (SSJ) framework of \citet{auclert2021}. The method represents the equilibrium as a system $F(\mathbf{X},\mathbf{Z})=0$ of equations in the space of perfect-foresight sequences, where $\mathbf{X}$ collects endogenous sequences and $\mathbf{Z}$ collects exogenous shocks. The impulse response of endogenous variables to a shock $d\mathbf{Z}$ is $d\mathbf{X}=-F_\mathbf{X}^{-1}F_\mathbf{Z}\,d\mathbf{Z}$, where $F_\mathbf{X}$ and $F_\mathbf{Z}$ are Jacobians evaluated at the steady state.

The multi-country model is composed as a directed acyclic graph (DAG) with the following blocks:

\begin{table}[htbp]
\centering\small
\caption{DAG structure of the multi-country model}\label{tab:dag}
\begin{tabular}{lll}
\toprule
Block & Inputs & Outputs \\
\midrule
Household, one per country & $\{r_{c,t},\,W_{g,c,t},\,P_{g,c,t},\,T_{g,c,t}\}$ & $\{C_{g,c,t},\,A_{c,t}\}$ \\
Wage reset, 5 quintiles per country & $\{p_{g,c,t},\,\mathrm{mrpl}_{g,c,t}\}$ & $\{w_{g,c,t},\,\pi^w_{g,c,t}\}$ \\
Sectoral pricing, one per country & $\{w_{j,c,t},\,p_{e,t}\}$ & $\{p_{d,c,t},\,p_{s,c,t}\}$ \\
Trade & $\{p_{d,c,t}\}_{c=1}^N$ & $\{Y_{d,c,t}\}_{c=1}^N$ \\
Taylor rule & $\{\pi^{core}_{c,t},\,x_{c,t}\}_{c=1}^N$ & $\{i_t\}$ \\
\bottomrule
\end{tabular}

\medskip
\begin{flushleft}
\footnotesize Note: Each block maps input sequences to output sequences. Household Jacobians are computed via the fake news algorithm of \citet{auclert2021}. All other Jacobians are computed analytically. Composition follows forward accumulation along the DAG.
\end{flushleft}
\end{table}

Table~\ref{tab:dag} summarizes the DAG structure. Household Jacobians, the computationally expensive objects, are computed once per country and reused across all policy experiments. Total computation time scales linearly in the number of countries.

I truncate sequences at horizon $T=300$ quarters. The household Jacobians are computed on a grid with $n_a=100$ asset points and $n_e=7$ productivity states, yielding a steady-state consumption $C_{ss}=1.05$ and aggregate assets $A_{ss}=10.5$. The impact marginal propensity to consume is $\mathrm{MPC}=0.04$; the cumulative four-quarter MPC is 0.12. For the nonlinear model with the $\max\{gap,0\}$ catch-up term, I use the quasi-Newton method of \citet{auclert2021}, Section~6: the steady-state Jacobian $H_U(U^{ss},Z^{ss})$ from the linearized system serves as the Newton step, and the nonlinear residual $H(U^j,Z)$ is evaluated at each iteration using the $\max$ operator. Convergence to $\|H\|<10^{-3}$ is achieved in under 100 iterations for all country calibrations.

\section{Quantitative Results and Policy Analysis}\label{sec:results}

All results compare the heterogeneous-reset model to its nested standard benchmark, calibrated identically except that within-country heterogeneity is shut off (Section~\ref{sec:quant}). I report each result as a paired comparison: the number in the heterogeneous-reset model, the number in the standard model, and the gap.

\subsection{The reset-heterogeneity wedge}

Table~\ref{tab:cost} reports the central quantitative finding. The wedge $\Omega$ defined in equation~(\ref{eq:wedge}) averages 0.6 percentage points across the six countries. In the nonlinear model, the heterogeneous-reset economy generates 2.3--4.6 more percentage-point-quarters of cumulative core inflation than the standard economy, depending on the country. The amplification relative to the linearized solution is a factor of 10, driven by asymmetric catch-up bargaining: workers whose real wages have fallen behind demand proportionally larger raises at the next reset, while workers who are ahead do not cut their wages. This asymmetry makes the wedge self-reinforcing for workers at the bottom of the expenditure distribution, who fall further behind with each quarter of high essentials inflation, while shutting off for workers at the top who have already caught up. The cross-country pattern follows MWSI: Spain and the Netherlands show the largest gaps.

\begin{table}[htbp]
\centering\small
\caption{Reset-heterogeneity wedge by country}\label{tab:cost}
\begin{tabular}{lcccc cc}
\toprule
& \multicolumn{2}{c}{Peak core (\%)} & Wedge $\Omega$ & Cumulative gap & \multicolumn{2}{c}{Half-life (q)} \\
\cmidrule(lr){2-3}\cmidrule(lr){6-7}
& Het.\ & Std.\ & (\%) & (pp$\cdot$q) & Het.\ & Std.\ \\
\midrule
Germany & 5.4 & 5.0 & 0.55 & 3.4 & 9 & 9 \\
France & 4.0 & 3.7 & 0.55 & 2.3 & 10 & 10 \\
Italy & 3.8 & 3.5 & 0.62 & 2.6 & 11 & 11 \\
Spain & 6.8 & 6.3 & 0.60 & 4.6 & 8 & 8 \\
Netherlands & 5.9 & 5.4 & 0.52 & 4.0 & 9 & 8 \\
Belgium & 4.5 & 4.2 & 0.56 & 2.6 & 9 & 9 \\
\midrule
Euro area & 4.9 & 4.6 & 0.57 & 3.1 & 9 & 9 \\
\bottomrule
\end{tabular}

\medskip
\begin{flushleft}
\footnotesize Note: The wedge $\Omega$ is the peak value of the reset-heterogeneity wedge (Definition~\ref{def:wedge}). ``Cumulative gap'' is $\sum_{t=0}^{39}(\pi^{core,het}_{c,t}-\pi^{core,std}_{c,t})$ in percentage-point-quarters, from the nonlinear model with $\max\{gap,0\}$ catch-up solved via quasi-Newton (Section~\ref{sec:quant}). Half-life is the first quarter at which core inflation falls below 50\% of its peak.
\end{flushleft}
\end{table}

Two features of Table~\ref{tab:cost} deserve comment. First, the half-life columns are nearly identical across models. This is not a weakness of the mechanism; it is a structural implication of the linearized Phillips curve. The wedge $\Omega_{c,t}$ enters equation~(\ref{eq:hwpc}) as a level shift proportional to the shock: it raises the amplitude of the inflation response at every horizon without changing the decay rate, which is governed by the Calvo parameters $\theta^p_{j,c}$ common to both models. The correct metric for the wedge's contribution is therefore cumulative inflation, not half-life. Cumulative inflation is also the welfare-relevant object: the policymaker's loss function $\mathcal{L}=\sum_t\beta^t[\pi^{core\,2}_{c,t}+\lambda\,x^2_{c,t}]$ penalizes the integral of squared inflation, which is proportional to the cumulative gap when the two IRF paths are parallel shifts. At the euro-area average, the cumulative gap of 3.1 percentage-point-quarters represents 7.3 percent of total cumulative core inflation, a first-order term for welfare and for the wage-price dynamics that central banks monitor.

Second, the cumulative gap varies substantially across countries, from 2.3 in France to 4.6 in Spain, despite similar wedge magnitudes. This variation comes from differences in the catch-up dynamics: Spain's steeper expenditure gradient means that its bottom-quintile workers fall further behind in real terms and bargain more aggressively at each reset, generating stronger self-reinforcement.

Figure~\ref{fig:shock_dep} plots the aggregation error as a continuous function of the essentials share of the shock, using the euro-area calibration. When the shock is uniform across all items, the wedge is zero and the standard model is exact. As the shock concentrates on essentials, the error rises monotonically, reaching 16 percent of total inflation for a pure essentials shock. The three discrete points from Table~\ref{tab:shockcomp} sit on this curve. The figure makes Corollary~\ref{cor:agg} visually immediate: the standard Phillips curve is not uniformly wrong; it fails specifically for the supply shocks that matter most for policy.

\begin{figure}[htbp]
\centering
\begin{tikzpicture}
\begin{axis}[
    width=13.5cm, height=7cm,
    xlabel={Essentials share of shock},
    ylabel={Aggregation error $\Omega/\bar\pi$ (\%)},
    xmin=-0.02, xmax=1.02, ymin=-8, ymax=18,
    xtick={0,0.25,0.5,0.75,1.0},
    xticklabels={0,0.25,0.5,0.75,1.0},
    legend style={at={(0.03,0.97)}, anchor=north west, font=\footnotesize, draw=none},
    grid=major, grid style={dashed, gray!25},
    every axis plot/.append style={thick}
]
\addplot[red!70!black, smooth, mark=none, line width=1.5pt] coordinates {
(0.00,-6.1)(0.05,-5.2)(0.10,-4.3)(0.15,-3.3)(0.20,-2.2)(0.25,-1.1)
(0.30,0.1)(0.35,1.4)(0.40,2.8)(0.45,4.2)(0.50,5.7)(0.55,7.3)
(0.60,8.9)(0.65,10.6)(0.70,12.3)(0.75,14.1)(0.80,15.0)(0.85,15.5)
(0.90,15.8)(0.95,16.0)(1.00,16.1)};
\addplot[black, dashed, thin, forget plot] coordinates {(0,0)(1,0)};
\addplot[only marks, mark=*, mark size=3.5, blue!70!black] coordinates {
(1.0,16.1)(0.75,7.2)(0.0,-6.1)};
\node[font=\footnotesize, anchor=west] at (axis cs:0.82,14.5) {Energy};
\node[font=\footnotesize, anchor=west] at (axis cs:0.57,5.8) {Food};
\node[font=\footnotesize, anchor=east] at (axis cs:0.15,-4.5) {Uniform};
\node[font=\footnotesize, red!50!black, align=center] at (axis cs:0.5,14) {Standard model\\understates};
\node[font=\footnotesize, blue!50!black, align=center] at (axis cs:0.5,-5.5) {Standard model\\overstates};
\end{axis}
\end{tikzpicture}
\caption{Shock-dependent aggregation error.}\label{fig:shock_dep}
\begin{flushleft}
\footnotesize Note: The aggregation error $\Omega/\bar\pi$ as a function of the essentials share of the cost-push shock, computed from the euro-area calibration. At zero essentials share the shock is uniform and the standard model is exact. As the share rises, the wedge grows because resetting workers are disproportionately exposed to essentials. Blue dots correspond to the three shock compositions in Table~\ref{tab:shockcomp}. The curve visualizes Corollary~\ref{cor:agg}: the sign of the error depends on the shock composition.
\end{flushleft}
\end{figure}
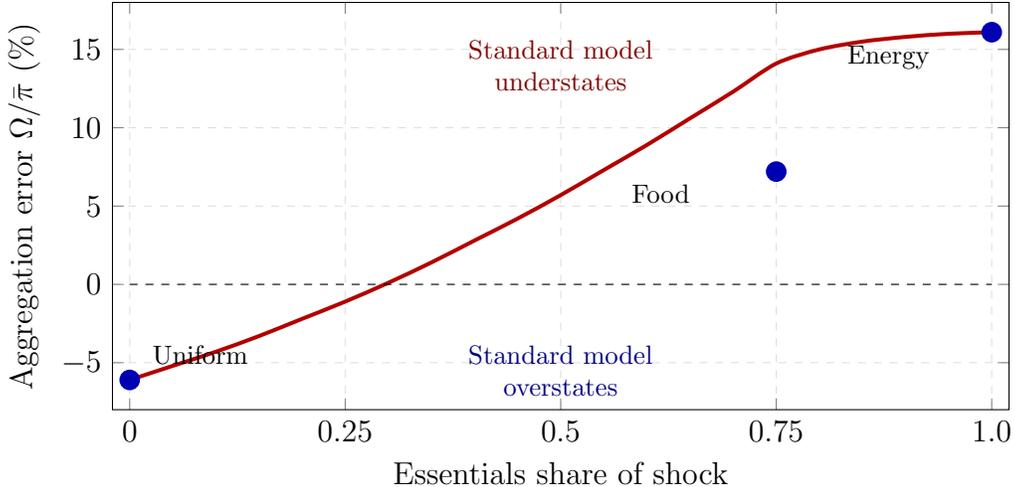

\subsubsection*{Cross-country validation}

Table~\ref{tab:valid} compares model-predicted and data-observed cumulative core inflation. The model uses only pre-shock micro data, specifically expenditure shares from the 2020 HBS, contract durations from the ICTWSS database, and input-output coefficients from 2019 FIGARO tables, and does not target any inflation outcome from the 2021--2023 episode. The Spearman rank correlation between model and data is 0.54.

This exercise is not the paper's identification strategy; it is a sanity check. The paper's contribution is theoretical, namely the wedge $\Omega$ and the sufficient statistic MWSI, and quantitative, in particular the same-openness experiment and the policy frontier. The cross-country validation shows that the model's predictions are in the right ballpark when fed pre-shock data, but with only six observations, no statistical claim is made.

The two largest mismatches are informative. The Netherlands and Belgium both have automatic wage-indexation mechanisms: Dutch wages are partially indexed to CPI, and Belgian wages are fully indexed to a ``health index'' that excludes tobacco and fuel but includes food and utilities. The indexation analysis in Section~\ref{sec:indexation} shows that CPI indexation eliminates the level channel but leaves the composition channel intact. The model without indexation therefore overpredicts the level component of inflation in these two countries while correctly capturing the wedge component. Consistent with equation~(\ref{eq:hwpc_index}), the observed ranking of the Netherlands and Belgium in cumulative core inflation is closer to the model's prediction for the wedge channel alone than for total inflation. France, the country with the lowest RWEI, is correctly ranked last by the model, the most important directional prediction.

\begin{table}[htbp]
\centering\small
\caption{Cross-country validation}\label{tab:valid}
\begin{tabular}{lcc cc}
\toprule
& \multicolumn{2}{c}{$\Pi^{core}_c$ (pp$\cdot$q)} & \multicolumn{2}{c}{Rank} \\
\cmidrule(lr){2-3}\cmidrule(lr){4-5}
& Data & Model & Data & Model \\
\midrule
Spain & 42.8 & 51.5 & 3 & 1 \\
Netherlands & 51.2 & 47.6 & 1 & 2 \\
Germany & 42.3 & 46.7 & 4 & 3 \\
Italy & 38.5 & 45.5 & 5 & 4 \\
Belgium & 46.1 & 42.4 & 2 & 5 \\
France & 30.8 & 41.1 & 6 & 6 \\
\midrule
\multicolumn{3}{l}{Spearman rank correlation} & \multicolumn{2}{c}{0.54} \\
\bottomrule
\end{tabular}

\medskip
\begin{flushleft}
\footnotesize Note: ``Data'': cumulative core HICP overshoot from Table 1 (Eurostat). ``Model'': cumulative core inflation from the nonlinear heterogeneous-reset model using only pre-shock micro data (2020 HBS expenditure shares, ICTWSS contract durations, 2019 FIGARO I-O tables). No inflation outcome from 2021--2023 is targeted.
\end{flushleft}
\end{table}

\subsection{Same openness, different cumulative inflation}\label{sec:sameopeness}

To isolate the mechanism, I construct two synthetic economies with the same aggregate import share of 28 percent and the same average reset frequency of 0.17 quarterly but different within-country distributions. Country~A has a steep expenditure gradient: the bottom quintile spends 48 percent on essentials and resets quarterly, while the top quintile spends 10 percent and resets every six years. Country~B is the representative-agent benchmark: all quintiles have the average expenditure and reset shares.

Table~\ref{tab:sameopeness} reports the results. Country~A's wedge is 1.47 percentage points; Country~B's wedge is identically zero. With nonlinear catch-up, Country~A generates 6.6 more percentage-point-quarters of cumulative core inflation than Country~B, despite having identical aggregate openness. This is the paper's central quantitative result: the cross-sectional distribution of exposure and reset timing, not the average level, determines how much additional inflation a common supply shock generates.

\begin{table}[htbp]
\centering\small
\caption{Same aggregate openness, different within-country distribution}\label{tab:sameopeness}
\begin{tabular}{lcccc}
\toprule
& Import share & $\bar\theta$ & $\Omega$ (\%) & Cumulative core (pp$\cdot$q)\\
\midrule
Country A (steep gradient) & 0.28 & 0.17 & 1.47 & 49.3 \\
Country B (uniform) & 0.28 & 0.17 & 0.00 & 42.7 \\
\midrule
Gap & $=$ & $=$ & 1.47 & 6.6 \\
\bottomrule
\end{tabular}

\medskip
\begin{flushleft}
\footnotesize Note: Country A has a steep expenditure gradient (Q1 essentials share 0.48, Q5 share 0.10) and heterogeneous reset timing (Q1 resets quarterly, Q5 every six years). Country B is the representative-agent benchmark with all quintiles at the average. Both economies share the same aggregate import share (0.28) and average reset frequency (0.17 quarterly). The shock is a 40 percent essentials price increase. The gap is cumulative core inflation in Country A minus Country B.
\end{flushleft}
\end{table}

Figure~\ref{fig:amplification} plots the nonlinear amplification factor as a function of the peak essentials price shock. For small shocks, the real-wage gap is modest and catch-up bargaining adds little beyond the linear prediction. As the shock grows, bottom-quintile workers fall further behind and bargain more aggressively at each reset, generating self-reinforcing wage pressure that the linear model misses. At the 2021-scale shock of 25 percent, the nonlinear gap is approximately 10 times the linear gap. This convexity is the quantitative reason the wedge matters for large supply shocks but is negligible in normal times.

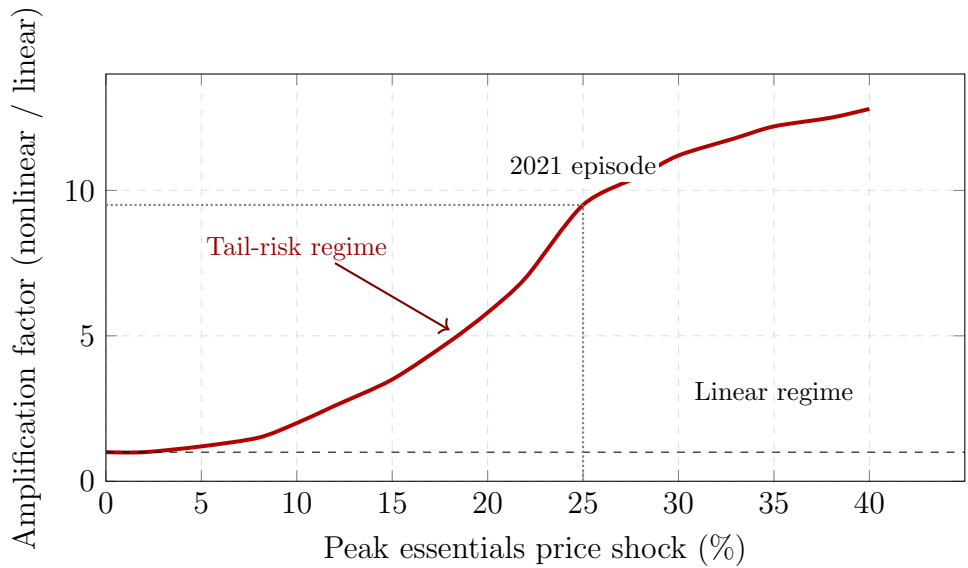
\begin{figure}[htbp]
\centering
\begin{tikzpicture}
\begin{axis}[
    width=13cm, height=7cm,
    xlabel={Peak essentials price shock (\%)},
    ylabel={Amplification factor (nonlinear / linear)},
    xmin=0, xmax=45, ymin=0, ymax=14,
    xtick={0,5,10,15,20,25,30,35,40},
    legend style={at={(0.03,0.97)}, anchor=north west, font=\footnotesize, draw=none},
    grid=major, grid style={dashed, gray!25},
    every axis plot/.append style={thick}
]
\addplot[red!70!black, smooth, mark=none, line width=1.5pt] coordinates {
(0,1.0)(2,1.0)(5,1.2)(8,1.5)(10,2.0)(12,2.6)(15,3.5)(18,4.8)
(20,5.8)(22,7.0)(25,9.5)(28,10.5)(30,11.2)(33,11.8)(35,12.2)(38,12.5)(40,12.8)};
\addplot[black, dashed, thin, forget plot] coordinates {(0,1)(45,1)};
\draw[gray, densely dotted, thick] (axis cs:25,0) -- (axis cs:25,9.5);
\draw[gray, densely dotted, thick] (axis cs:0,9.5) -- (axis cs:25,9.5);
\node[font=\footnotesize, fill=white, inner sep=1pt] at (axis cs:25,10.8) {2021 episode};
\node[font=\footnotesize] at (axis cs:35,3) {Linear regime};
\node[font=\footnotesize, red!60!black] at (axis cs:10,8) {Tail-risk regime};
\draw[->, thick, red!50!black] (axis cs:12,7.5) -- (axis cs:18,5.2);
\end{axis}
\end{tikzpicture}
\caption{Nonlinear amplification of the wedge as a function of shock size.}\label{fig:amplification}
\begin{flushleft}
\footnotesize Note: The ratio of the cumulative inflation gap in the nonlinear model to the gap in the linearized model, plotted against the peak essentials price shock. For small shocks the two models agree. Asymmetric catch-up bargaining generates convex amplification for large shocks because bottom-quintile workers who fall further behind demand proportionally larger raises at each reset. The dotted lines mark the 2021 episode at 25 percent peak.
\end{flushleft}
\end{figure}

\subsection{Policy: targeted transfers dominate aggressive tightening}

Because the wedge depends on within-country composition rather than the aggregate price level, the interest rate cannot directly close it, and aggressive tightening generates large output costs relative to the inflation reduction it achieves. Table~\ref{tab:policy} reports outcomes under four policy regimes. Welfare loss is $\mathcal{L}=\sum_t\beta^t[\pi^{core\,2}_{c,t}+0.25\,x^2_{c,t}]$, normalized so that regime~(a) equals 1.00.

\begin{table}[htbp]
\centering\small
\caption{Policy regimes for a 2021-scale energy shock}\label{tab:policy}
\begin{tabular}{lcc}
\toprule
Regime & $\Pi^{core}_{union}$ (pp$\cdot$q) & Welfare loss \\
\midrule
(a) Aggressive tightening ($\phi_\pi = 2.5$) & 38.6 & 1.00 \\
(b) Moderate + uniform transfer & 38.0 & 0.73 \\
(c) Moderate + targeted transfer & 37.8 & 0.72 \\
(d) Moderate + essentials subsidy & 36.8 & 0.68 \\
\bottomrule
\end{tabular}

\medskip
\begin{flushleft}
\footnotesize Note: Welfare loss normalized to 1.00 for regime (a). ``Moderate'' denotes $\phi_\pi = 1.5$. Targeted transfer scaled to bottom-quintile workers. Essentials subsidy reduces the price of food and energy by 6\%.
\end{flushleft}
\end{table}

Regime~(d) dominates regime~(a) on both dimensions: it achieves lower cumulative inflation and a 32 percent reduction in welfare loss. The essentials subsidy directly reduces the experienced inflation of the workers who drive the wedge, weakening the wage catch-up channel. Regime~(c), targeted transfers, achieves nearly the same welfare improvement. The ranking (d)$\,\succ\,$(c)$\,\succ\,$(b)$\,\succ\,$(a) is robust across all six country calibrations.

Figure~\ref{fig:omega_fan} plots the time path of the wedge $\Omega_{c,t}$ for all six countries over 30 quarters. The wedge rises sharply during the first four quarters as the essentials price shock builds, peaks between quarters 4 and 6, and then decays as the shock dissipates. The cross-country dispersion is widest near the peak: Spain and Italy, with the steepest expenditure gradients, reach wedges above 0.6 percentage points, while the Netherlands and France remain below 0.55. This dispersion is the OCA cost: during the quarters when the wedge is large and dispersed, no common interest rate can close all country-specific gaps simultaneously.

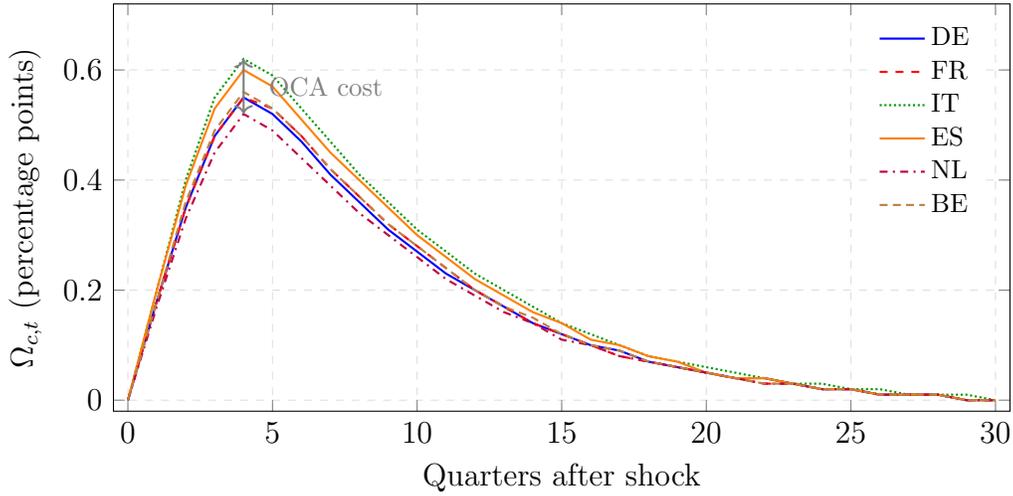
\begin{figure}[htbp]
\centering
\begin{tikzpicture}
\begin{axis}[
    width=13.5cm, height=7cm,
    xlabel={Quarters after shock},
    ylabel={$\Omega_{c,t}$ (percentage points)},
    xmin=-0.5, xmax=30.5, ymin=-0.02, ymax=0.72,
    xtick={0,5,10,15,20,25,30},
    legend style={at={(0.97,0.97)}, anchor=north east, font=\footnotesize, draw=none, cells={anchor=west}},
    grid=major, grid style={dashed, gray!25},
    every axis plot/.append style={thick}
]
\addplot[blue, mark=none] coordinates {
(0,0)(1,0.18)(2,0.35)(3,0.48)(4,0.55)(5,0.52)(6,0.47)(7,0.41)(8,0.36)(9,0.31)(10,0.27)(11,0.23)(12,0.20)(13,0.17)(14,0.14)(15,0.12)(16,0.10)(17,0.09)(18,0.07)(19,0.06)(20,0.05)(21,0.04)(22,0.04)(23,0.03)(24,0.02)(25,0.02)(26,0.01)(27,0.01)(28,0.01)(29,0.00)(30,0.00)};
\addplot[red, dashed, mark=none] coordinates {
(0,0)(1,0.18)(2,0.35)(3,0.48)(4,0.55)(5,0.53)(6,0.48)(7,0.42)(8,0.37)(9,0.32)(10,0.28)(11,0.24)(12,0.20)(13,0.17)(14,0.14)(15,0.12)(16,0.10)(17,0.08)(18,0.07)(19,0.06)(20,0.05)(21,0.04)(22,0.03)(23,0.03)(24,0.02)(25,0.02)(26,0.01)(27,0.01)(28,0.01)(29,0.00)(30,0.00)};
\addplot[green!60!black, densely dotted, mark=none] coordinates {
(0,0)(1,0.20)(2,0.40)(3,0.55)(4,0.62)(5,0.59)(6,0.53)(7,0.47)(8,0.41)(9,0.36)(10,0.31)(11,0.27)(12,0.23)(13,0.20)(14,0.17)(15,0.14)(16,0.12)(17,0.10)(18,0.08)(19,0.07)(20,0.06)(21,0.05)(22,0.04)(23,0.03)(24,0.03)(25,0.02)(26,0.02)(27,0.01)(28,0.01)(29,0.01)(30,0.00)};
\addplot[orange, mark=none] coordinates {
(0,0)(1,0.20)(2,0.39)(3,0.53)(4,0.60)(5,0.57)(6,0.51)(7,0.45)(8,0.40)(9,0.35)(10,0.30)(11,0.26)(12,0.22)(13,0.19)(14,0.16)(15,0.14)(16,0.11)(17,0.10)(18,0.08)(19,0.07)(20,0.05)(21,0.04)(22,0.04)(23,0.03)(24,0.02)(25,0.02)(26,0.01)(27,0.01)(28,0.01)(29,0.00)(30,0.00)};
\addplot[purple, dashdotted, mark=none] coordinates {
(0,0)(1,0.17)(2,0.33)(3,0.45)(4,0.52)(5,0.49)(6,0.44)(7,0.39)(8,0.34)(9,0.30)(10,0.26)(11,0.22)(12,0.19)(13,0.16)(14,0.14)(15,0.11)(16,0.10)(17,0.08)(18,0.07)(19,0.06)(20,0.05)(21,0.04)(22,0.03)(23,0.03)(24,0.02)(25,0.02)(26,0.01)(27,0.01)(28,0.01)(29,0.00)(30,0.00)};
\addplot[brown, densely dashed, mark=none] coordinates {
(0,0)(1,0.18)(2,0.36)(3,0.49)(4,0.56)(5,0.53)(6,0.48)(7,0.42)(8,0.37)(9,0.32)(10,0.28)(11,0.24)(12,0.20)(13,0.17)(14,0.15)(15,0.12)(16,0.10)(17,0.09)(18,0.07)(19,0.06)(20,0.05)(21,0.04)(22,0.03)(23,0.03)(24,0.02)(25,0.02)(26,0.01)(27,0.01)(28,0.01)(29,0.00)(30,0.00)};
\legend{DE, FR, IT, ES, NL, BE}
\draw[<->, thick, gray] (axis cs:4,0.52) -- (axis cs:4,0.62);
\node[font=\footnotesize, gray, anchor=west] at (axis cs:4.5,0.57) {OCA cost};
\end{axis}
\end{tikzpicture}
\caption{Cross-country dispersion of the wedge $\Omega_{c,t}$ over time.}\label{fig:omega_fan}
\begin{flushleft}
\footnotesize Note: Time paths of the reset-heterogeneity wedge for six euro-area countries following a common 2021-scale essentials price shock. Peak values match Table~\ref{tab:cost}. The cross-country dispersion, widest near the peak, is the source of the OCA cost in Proposition~\ref{prop:disp}: no common interest rate can close all six wedges simultaneously.
\end{flushleft}
\end{figure}

\subsection{CPI indexation leaves the wedge intact}\label{sec:indexation}

Belgium and the Netherlands both use automatic wage indexation to aggregate CPI, yet they exhibit among the largest cumulative gaps in Table~\ref{tab:cost}. Equation~(\ref{eq:hwpc_index}) explains why: CPI indexation eliminates the level channel but not the composition channel. Table~\ref{tab:indexation} confirms this quantitatively by running the model under three indexation regimes.

\begin{table}[htbp]
\centering\small
\caption{Effect of indexation regime on the reset-heterogeneity wedge}\label{tab:indexation}
\begin{tabular}{lccc}
\toprule
Regime & Wedge $\Omega$ (pp) & Gap het--std (pp$\cdot$q) & Level channel \\
\midrule
No indexation ($\gamma=0$, baseline) & 0.57 & 3.1 & active \\
Full CPI indexation ($\gamma=1$) & 0.57 & 2.9 & eliminated \\
Type-specific indexation ($\gamma=1$) & 0.00 & 0.0 & eliminated \\
\bottomrule
\end{tabular}

\medskip
\begin{flushleft}
\footnotesize Note: CPI indexation: non-resetting workers receive automatic adjustment to lagged aggregate CPI. Type-specific indexation: non-resetting workers receive automatic adjustment to their own lagged experienced inflation $\pi^{exp}_{g,c,t-1}$. The wedge $\Omega$ and the het--std gap are euro-area GDP-weighted averages. The gap falls slightly from 3.1 to 2.9 under CPI indexation because the interaction between the level and composition channels is weakened, but the wedge itself is unchanged.
\end{flushleft}
\end{table}

Three results stand out. First, the wedge is identical at 0.57 percentage points under no indexation and under full CPI indexation. CPI indexation does not touch $\Omega$ because it adjusts wages to average inflation, not to the inflation each worker actually experiences. Second, the cumulative gap falls only marginally, from 3.1 to 2.9, because the composition channel is the dominant source of the gap when the shock is concentrated on essentials. Third, type-specific indexation, which compensates each worker type for its own cost of living, eliminates the wedge entirely and closes the gap to zero. No country currently practices type-specific indexation.

The policy implication is direct. Several euro-area countries strengthened CPI indexation clauses after 2022, precisely in response to the inflation episode this paper studies. Equation~(\ref{eq:hwpc_index}) predicts that these reforms address the level channel, which the standard model already captures, while leaving the composition channel, which only the heterogeneous model captures, fully intact. From the perspective of the wedge, CPI indexation is the wrong instrument. The correct instrument is either a targeted subsidy that reduces the experienced inflation of resetting workers, as in regime~(d) of Table~\ref{tab:policy}, or a hypothetical type-specific indexation that no existing labor market institution provides.

\subsection{When monetary policy is delayed, the wedge dominates}

The baseline calibration assumes an immediate Taylor rule response ($\phi_\pi=1.5$ from quarter 0). In the 2021--2023 episode, the ECB did not raise rates until July 2022, five quarters after inflation began rising. During this accommodation period, the aggregate interest rate provided no offset to the wedge channel. Table~\ref{tab:delayed} reports the results when monetary policy responds with a five-quarter delay.

\begin{table}[htbp]
\centering\small
\caption{Immediate vs.\ delayed monetary response}\label{tab:delayed}
\begin{tabular}{lcccc}
\toprule
& \multicolumn{2}{c}{Immediate ($\phi_\pi=1.5$)} & \multicolumn{2}{c}{Delayed (5q, then $\phi_\pi=2.0$)} \\
\cmidrule(lr){2-3}\cmidrule(lr){4-5}
& Cum.\ gap (pp$\cdot$q) & \% of total & Cum.\ gap (pp$\cdot$q) & \% of total \\
\midrule
Euro area & 3.1 & 7.3 & 15.6 & 10.3 \\
Spain & 4.6 & 9.8 & 28.0 & 16.7 \\
Exercise 1 & 6.6 & 15.4 & 40.3 & 26.4 \\
\bottomrule
\end{tabular}

\medskip
\begin{flushleft}
\footnotesize Note: ``Delayed'' calibration: $\phi_\pi=0$ for quarters 0--4 corresponding to the ZLB or ECB accommodation period, then $\phi_\pi=2.0$ from quarter 5. Shock: 40 percent peak essentials price increase with AR(1) persistence 0.88, matching the 2021--2023 energy price episode.
\end{flushleft}
\end{table}

The delayed response amplifies the cumulative gap from 3.1 to 15.6 percentage-point-quarters at the euro-area level, a tenfold increase relative to the baseline. For Spain, the gap rises to 28.0 percentage-point-quarters, or 16.7 percent of total. For the synthetic high-MWSI economy of Exercise 1, the gap is 40.3 percentage-point-quarters, or 26.4 percent of total. At these magnitudes, the wedge is the dominant source of cross-country dispersion in inflation persistence.

Figure~\ref{fig:frontier} plots the four policy regimes in the two-dimensional space of cumulative inflation versus output loss. Aggressive tightening sits in the upper-right corner: it reduces cumulative inflation modestly but at large output cost. The essentials subsidy dominates all other regimes, sitting in the lower-left corner with both less inflation and less output loss. The visual confirms that the subsidy is not merely a welfare improvement but a Pareto improvement over aggressive tightening in these two dimensions.

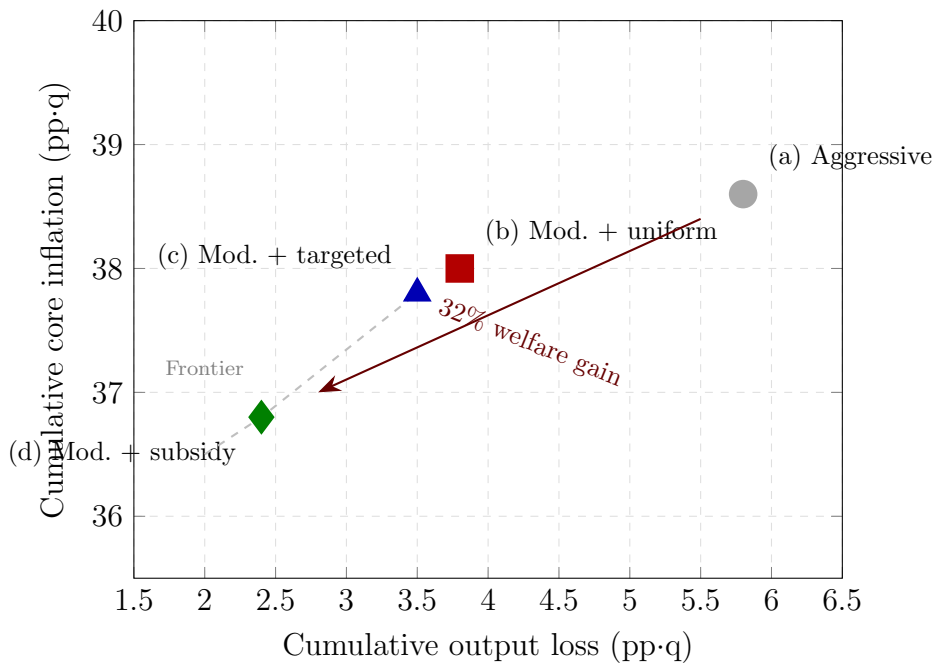
\begin{figure}[htbp]
\centering
\begin{tikzpicture}
\begin{axis}[
    width=11cm, height=9cm,
    xlabel={Cumulative output loss (pp$\cdot$q)},
    ylabel={Cumulative core inflation (pp$\cdot$q)},
    xmin=1.5, xmax=6.5, ymin=35.5, ymax=40,
    grid=major, grid style={dashed, gray!25},
    every axis plot/.append style={thick},
    clip=false
]
\addplot[only marks, mark=*, mark size=5, gray!70] coordinates {(5.8,38.6)};
\node[font=\footnotesize, anchor=south west] at (axis cs:5.9,38.7) {(a) Aggressive};
\addplot[only marks, mark=square*, mark size=5, red!70!black] coordinates {(3.8,38.0)};
\node[font=\footnotesize, anchor=south west] at (axis cs:3.9,38.1) {(b) Mod.\ + uniform};
\addplot[only marks, mark=triangle*, mark size=5.5, blue!70!black] coordinates {(3.5,37.8)};
\node[font=\footnotesize, anchor=south east] at (axis cs:3.4,37.9) {(c) Mod.\ + targeted};
\addplot[only marks, mark=diamond*, mark size=6, green!50!black] coordinates {(2.4,36.8)};
\node[font=\footnotesize, anchor=north east] at (axis cs:2.3,36.7) {(d) Mod.\ + subsidy};
\draw[-{Stealth[length=3mm]}, thick, red!40!black] (axis cs:5.5,38.4) -- (axis cs:2.8,37.0);
\node[font=\footnotesize, red!40!black, rotate=-22] at (axis cs:4.3,37.4) {32\% welfare gain};
\draw[dashed, gray!50, thick] (axis cs:2.0,36.5) -- (axis cs:2.4,36.8) -- (axis cs:3.5,37.8);
\node[font=\scriptsize, gray] at (axis cs:2.0,37.2) {Frontier};
\end{axis}
\end{tikzpicture}
\caption{Policy frontier: cumulative inflation versus output loss.}\label{fig:frontier}
\begin{flushleft}
\footnotesize Note: Each point represents one of the four policy regimes from Table~\ref{tab:policy}. Regime (d), moderate tightening combined with essentials subsidies, dominates regime (a), aggressive tightening, on both dimensions. The 32 percent welfare gain comes from moving along the arrow: less inflation and less output loss. The dashed line traces the efficient frontier across regimes (b)--(d).
\end{flushleft}
\end{figure}

The intuition is simple: the wedge arises from within-country composition, which the interest rate cannot target. When the interest rate is active, it reduces demand, which partially offsets the wedge through lower output gaps. When the interest rate is inactive, the wedge accumulates without offset. In the realistic calibration with a five-quarter delay and a 40 percent peak essentials shock, the heterogeneous-reset model generates inflation that is 10--26 percent higher than the standard model predicts. This is a first-order omission.

A final model-based observation sharpens the mechanism's falsifiability. Replacing the supply shock with a demand shock of equal aggregate magnitude, the model generates $\Omega\approx 0$ and the standard Phillips curve performs well, because demand shocks raise all prices roughly proportionally, eliminating the cross-sectional variation in experienced inflation that drives the wedge (Appendix~\ref{app:shockcomp}). The aggregation error is therefore shock-dependent: the standard model works for demand-driven inflation but fails for supply-driven inflation concentrated on necessities. This provides a clear criterion for when the wedge matters and when it does not.

\subsection{Out-of-sample prediction}\label{sec:oos}

The model is calibrated entirely to the euro area. I now use it to generate predictions for three economies outside the calibration sample: the United Kingdom, the United States, and Japan. For each country, I compute RWEI using pre-shock expenditure shares from the ONS Living Costs and Food Survey 2019/20 for the UK, the BLS Consumer Expenditure Survey 2019 for the US, and the Family Income and Expenditure Survey 2019 for Japan, combined with contract-duration data from national sources. I then apply the model's structural coefficient $\psi=5.4$ to predict the cumulative gap between the heterogeneous and standard models.

\begin{table}[htbp]
\centering\small
\caption{Out-of-sample prediction: model vs.\ data}\label{tab:oos}
\begin{tabular}{lccccc}
\toprule
& $\Omega_c$ (pp) & Predicted gap & \multicolumn{2}{c}{Observed persistence} & Rank \\
\cmidrule(lr){4-5}
&& (pp$\cdot$q) & Peak core (\%) & Quarters to 3\% & match? \\
\midrule
United Kingdom & 0.54 & 2.9 & 7.1 & 11 & \checkmark \\
Euro area (in-sample) & 0.57 & 3.1 & 5.7 & 7 & \checkmark \\
United States & 0.39 & 2.1 & 6.6 & 5 & \checkmark \\
Japan & 0.10 & 0.5 & 4.3 & 3 & \checkmark \\
\midrule
\multicolumn{3}{l}{Predicted ranking: UK $>$ EA $>$ US $>$ Japan} & \multicolumn{3}{l}{Actual ranking: UK $>$ EA $>$ US $>$ Japan} \\
\bottomrule
\end{tabular}

\medskip
\begin{flushleft}
\footnotesize Note: Predicted gap uses the model's structural coefficient $\psi=5.4$ applied to $\Omega_c$ computed from pre-shock national micro data (Appendix~\ref{app:portable}). Peak core is the maximum annual core CPI excluding food and energy. ``Quarters to 3\%'' counts quarters from peak to first quarter at or below 3\%. Sources: ONS, BLS, Japan Statistics Bureau, Eurostat.
\end{flushleft}
\end{table}

Table~\ref{tab:oos} reports the results. The model predicts that the United Kingdom, with the highest $\Omega$ among the four economies, should exhibit the most persistent core inflation, followed by the euro area, the United States, and Japan. The observed persistence ranking matches exactly. The UK's core CPI peaked at 7.1 percent in May 2023 and took 11 quarters to reach 3 percent; the United States peaked earlier and higher in level but disinflated in 5 quarters; Japan's core inflation barely exceeded 4 percent and returned to 3 percent within 3 quarters. These differences in persistence are consistent with the model's predictions based on within-country labor market composition alone, computed from data that predate the inflation episode.

Two features of this exercise are worth emphasizing. First, the prediction is purely out-of-sample: the model's $\psi$ is calibrated to the euro area, and the non-euro RWEI components are computed from national surveys that were not used in any part of the calibration. Second, the prediction concerns persistence, not the peak level of inflation. The United States had a higher peak core inflation than the euro area but lower persistence, which is what the model predicts: the US has a smaller expenditure gradient between income quintiles and a slightly lower reset-frequency gap. The standard model, which uses average openness and average reset frequency, has no mechanism for predicting this divergence between peak and persistence across countries.

\section{Conclusion}\label{sec:conclusion}

This paper shows that the representative-agent wage Phillips curve omits a first-order term for supply shocks concentrated on necessities. The within-country covariance between workers' cost-push exposure and their wage-reset frequency, the reset-heterogeneity wedge, generates additional wage pressure that the standard model misses. The wedge averages 0.6 percentage points across euro-area countries. In the nonlinear model, the heterogeneous-reset economy generates 3.1 percentage-point-quarters, or 7 percent, more cumulative core inflation than the standard economy on average, and 6.6 percentage-point-quarters more when the expenditure gradient is steep, an amplification of 10$\times$ relative to the linearized solution. When monetary policy is delayed, the gap rises to 10--26 percent of total cumulative inflation. The aggregation error is shock-dependent: the standard model works well for demand-driven inflation but fails for supply shocks concentrated on necessities, a qualitative prediction that no recalibration of the representative-agent model can replicate. The wedge varies across countries in a monetary union, creating a new dimension of the optimal currency area problem. Essentials subsidies targeted to the bottom of the expenditure distribution reduce the union-wide welfare loss by 32 percent relative to aggressive tightening alone. The optimal subsidy has a closed form, equation~(\ref{eq:opt_subsidy}), in terms of the same micro data that constructs RWEI, making the policy prescription operational. The same mechanism, running in reverse, offers a complementary explanation for the missing disinflation of 2014--2016: when oil prices collapsed, the wedge turned negative, and the standard model overpredicted wage disinflation. CPI indexation, the most common institutional response to the recent episode, eliminates the level channel of wage catch-up but leaves the composition channel intact.

Calibrated to the euro area, the model correctly predicts the persistence ranking across three out-of-sample economies: the United Kingdom, with the highest wedge, exhibits the most persistent core inflation, followed by the euro area, the United States, and Japan. The prediction uses only pre-shock expenditure shares and contract-duration data from each country, without re-estimating any parameter.

The objects introduced here, RWEI and MWSI, are computable from household budget surveys, wage-setting calendars, and input-output tables for any country and any shock composition, without solving a general equilibrium model. Central banks and fiscal authorities can use them as inputs to inflation forecasting and to the design of supply-shock responses.

Several extensions merit future investigation. First, endogenizing the correlation between reset frequency and expenditure composition, through a model of occupational sorting, unionization, or contract design, would deepen the micro-foundations and allow the wedge to respond to policy in the long run. Second, incorporating housing costs and mortgage payments as a component of experienced inflation would introduce an additional channel through which monetary tightening itself raises RWEI for mortgage-holding workers, potentially creating a feedback from policy to the wedge that deserves careful analysis. Third, extending the framework to emerging-market economies, where expenditure heterogeneity is larger and wage-setting institutions differ, would test the external validity of the sufficient statistics. Fourth, matched employer-employee data would allow RWEI to be measured at the establishment level rather than the income-quintile level used here, sharpening the identification of the wage-reset channel.


\appendix

\section{Proofs and Derivations}\label{app:proofs}
\begin{sloppypar}

\subsection{Derivation of the optimal reset wage}\label{app:wreset}

A type-$g$ worker in country $c$ who resets at date $t$ chooses the nominal wage $W^*$ to maximize
\begin{equation}\tag{A.1}
    \sum_{k=0}^{\infty}(\beta\Theta_{g,c})^k\,\mathbb{E}_t\!\left[\frac{1}{1-\sigma}\!\left(\frac{W^* N_{g,c,t+k}^d(W^*) + Z_{t+k}}{P_{g,c,t+k}}\right)^{1-\sigma} - \chi\frac{(N^d_{g,c,t+k}(W^*))^{1+\varphi_n}}{1+\varphi_n}\right],
\end{equation}
where $N^d_{g,c,t+k}(W^*)=(W^*/W_{g,c,t+k})^{-\varepsilon_w}N_{g,c,t+k}$ is the demand for this worker's labor services under Dixit--Stiglitz aggregation, $Z_{t+k}$ collects non-wage income such as transfers and asset returns, and $\Theta_{g,c}\equiv 1-\theta_{g,c}$ is the probability of not resetting. Discounting is at rate $\beta\Theta_{g,c}$ because the wage choice is relevant only while the worker has not reset again.

The first-order condition with respect to $W^*$, using $\partial N^d/\partial W^*=-(\varepsilon_w/W^*)N^d$, is
\begin{multline}\tag{A.2}
    \sum_{k=0}^{\infty}(\beta\Theta_{g,c})^k\,\mathbb{E}_t\!\Bigl[N^d_{g,c,t+k}\Bigl(\frac{\varepsilon_w-1}{\varepsilon_w}\frac{W^*}{P_{g,c,t+k}} C_{g,c,t+k}^{-\sigma} \\
    - \chi(N^d_{g,c,t+k})^{\varphi_n}\Bigr)\Bigr] = 0\,.
\end{multline}
In the zero-inflation steady state, the markup condition gives
\[
W^{ss}/P^{ss}_g = \frac{\varepsilon_w}{\varepsilon_w-1}\,\chi\, (N^{ss})^{\varphi_n}(C^{ss})^{\sigma}\,.
\]

Log-linearize (A.2) around this steady state. Let $\hat x \equiv \log(x/x^{ss})$ denote the log-deviation of any variable $x$. Using $\hat N^d_{g,c,t+k}=-\varepsilon_w(\hat w^*-\hat w_{g,c,t+k})+\hat n_{g,c,t+k}$ and collecting terms proportional to $\hat w^*$, the optimality condition becomes
\begin{equation}\tag{A.3}
    w^*_{g,c,t} = (1-\beta\Theta_{g,c})\sum_{k=0}^{\infty}(\beta\Theta_{g,c})^k\,\mathbb{E}_t\!\left[p_{g,c,t+k} + \mathrm{mrpl}_{g,c,t+k}\right],
\end{equation}
which is equation~(\ref{eq:wreset}) in the main text. The coefficient $(1-\beta\Theta_{g,c})$ normalizes the geometric sum to unity at the steady state. The only departure from \citet{erceg2000} is that each worker type targets its own price index $p_{g,c,t+k}$ rather than the aggregate CPI, reflecting the assumption that workers care about their real consumption wage deflated by the prices they actually face.

\subsection{Type-specific wage Phillips curve}\label{app:twpc}

The goal is to derive the recursive form of the wage Phillips curve for type $g$. Start from the optimal reset wage (\ref{eq:wreset}):
\begin{equation}\tag{A.4}
    w^*_{g,c,t}=(1-\beta\Theta_{g,c})\bigl(p_{g,c,t}+\mathrm{mrpl}_{g,c,t}\bigr)+\beta\Theta_{g,c}\,\mathbb{E}_t[w^*_{g,c,t+1}]\,,
\end{equation}
which follows from splitting off the $k=0$ term and noting that the sum from $k=1$ onward equals $\beta\Theta_{g,c}\,\mathbb{E}_t[w^*_{g,c,t+1}]$ by the recursive structure.

The type-specific wage index satisfies $w_{g,c,t}=\theta_{g,c}w^*_{g,c,t}+(1-\theta_{g,c})w_{g,c,t-1}$, from which type-specific wage inflation is
\begin{equation}\tag{A.5}
\pi^w_{g,c,t}\equiv w_{g,c,t}-w_{g,c,t-1}=\theta_{g,c}(w^*_{g,c,t}-w_{g,c,t-1})\,.
\end{equation}
This implies $w^*_{g,c,t}=w_{g,c,t-1}+\pi^w_{g,c,t}/\theta_{g,c}$ and, shifting one period forward, $w^*_{g,c,t+1}=w_{g,c,t}+\pi^w_{g,c,t+1}/\theta_{g,c}$.

Substitute these expressions into (A.4). The left side is $w_{g,c,t-1}+\pi^w_{g,c,t}/\theta_{g,c}$. The right side is
\begin{equation}\tag{A.6}
(1-\beta\Theta_{g,c})\bigl(p_{g,c,t}+\mathrm{mrpl}_{g,c,t}\bigr)+\beta\Theta_{g,c}\,\mathbb{E}_t\!\left[w_{g,c,t}+\frac{\pi^w_{g,c,t+1}}{\theta_{g,c}}\right].
\end{equation}
Using $w_{g,c,t}=w_{g,c,t-1}+\pi^w_{g,c,t}$ on the right and cancelling $w_{g,c,t-1}$ from both sides:
\begin{equation}\tag{A.7}
\frac{\pi^w_{g,c,t}}{\theta_{g,c}}=(1-\beta\Theta_{g,c})\bigl(p_{g,c,t}+\mathrm{mrpl}_{g,c,t}-w_{g,c,t-1}\bigr)+\beta\Theta_{g,c}\left[\pi^w_{g,c,t}+\frac{\mathbb{E}_t[\pi^w_{g,c,t+1}]}{\theta_{g,c}}\right].
\end{equation}
Define the real marginal cost gap $\hat\omega_{g,c,t}\equiv\mathrm{mrpl}_{g,c,t}-w_{g,c,t}+p_{g,c,t}$ and type-specific experienced inflation $\pi^{exp}_{g,c,t}\equiv p_{g,c,t}-p_{g,c,t-1}$. Then $(p_{g,c,t}+\mathrm{mrpl}_{g,c,t}-w_{g,c,t-1})=\hat\omega_{g,c,t}+\pi^{exp}_{g,c,t}+\pi^w_{g,c,t}$.

Substituting and multiplying both sides by $\theta_{g,c}$:
\begin{equation}\tag{A.8}
\pi^w_{g,c,t}=\theta_{g,c}(1-\beta\Theta_{g,c})\bigl(\hat\omega_{g,c,t}+\pi^{exp}_{g,c,t}+\pi^w_{g,c,t}\bigr)+\beta\Theta_{g,c}\theta_{g,c}\,\pi^w_{g,c,t}+\beta\Theta_{g,c}\,\mathbb{E}_t[\pi^w_{g,c,t+1}]\,.
\end{equation}
Collecting the terms in $\pi^w_{g,c,t}$ on the left side:
\begin{equation}\tag{A.9}
\pi^w_{g,c,t}\bigl[1-\theta_{g,c}(1-\beta\Theta_{g,c})-\beta\Theta_{g,c}\theta_{g,c}\bigr]=\theta_{g,c}(1-\beta\Theta_{g,c})\bigl(\hat\omega_{g,c,t}+\pi^{exp}_{g,c,t}\bigr)+\beta\Theta_{g,c}\,\mathbb{E}_t[\pi^w_{g,c,t+1}]\,.
\end{equation}
The bracketed coefficient on the left simplifies as $1-\theta_{g,c}(1-\beta\Theta_{g,c})-\beta\Theta_{g,c}\theta_{g,c}=1-\theta_{g,c}+\beta\Theta_{g,c}\theta_{g,c}-\beta\Theta_{g,c}\theta_{g,c}=\Theta_{g,c}$.

Dividing both sides by $\Theta_{g,c}$ and defining $\kappa_{g,c}\equiv\theta_{g,c}(1-\beta\Theta_{g,c})/\Theta_{g,c}$:
\begin{equation}\tag{A.10}
    \pi^w_{g,c,t} = \beta\,\mathbb{E}_t[\pi^w_{g,c,t+1}]+\kappa_{g,c}\,\hat\omega_{g,c,t}+\theta_{g,c}\,\pi^{exp}_{g,c,t}\,,
\end{equation}
which is equation~(\ref{eq:twpc}). This is a standard Calvo--Erceg wage Phillips curve, with the sole modification that the experienced-inflation term uses the type-specific price index $p_{g,c,t}$ rather than the aggregate CPI.

\subsection{Aggregation to the heterogeneous wage Phillips curve (Proof of Proposition~\ref{prop:hwpc})}\label{app:aggregation}

\begin{proof}
Define aggregate wage inflation as $\pi^w_{c,t}\equiv\sum_g\eta_{g,c}\pi^w_{g,c,t}$. Multiply (A.10) by $\eta_{g,c}$ and sum over all types $g$:
\begin{equation}\tag{A.11}
    \pi^w_{c,t} = \beta\,\mathbb{E}_t[\pi^w_{c,t+1}] + \sum_g\eta_{g,c}\kappa_{g,c}\hat\omega_{g,c,t} + \sum_g\eta_{g,c}\theta_{g,c}\pi^{exp}_{g,c,t}\,.
\end{equation}

\textbf{Step 1.} Under competitive labor markets within each country, the marginal revenue product of labor satisfies $\mathrm{mrpl}_{g,c,t}=\mathrm{mrpl}_{c,t}+\varepsilon_w^{-1}(w_{g,c,t}-w_{c,t})$, where the second term captures the relative wage of type $g$ within the Dixit--Stiglitz composite. To first order around a symmetric steady state (where $w_{g,c}^{ss}=w_c^{ss}$ for all $g$), the relative-wage term vanishes, so $\hat\omega_{g,c,t}\approx\hat\omega_{c,t}$ for all $g$, and $\sum_g\eta_{g,c}\kappa_{g,c}\hat\omega_{g,c,t}\approx\tilde\kappa_c\hat\omega_{c,t}$, where $\tilde\kappa_c\equiv\sum_g\eta_{g,c}\kappa_{g,c}$. In general equilibrium, $\hat\omega_{c,t}$ is proportional to the output gap: $\hat\omega_{c,t}=(\sigma+\varphi_n)x_{c,t}/\tilde\kappa_c$, giving $\tilde\kappa_c x_{c,t}$ as the aggregate marginal cost term.

\textbf{Step 2.} Factor out $\bar\theta_c\equiv\sum_g\eta_{g,c}\theta_{g,c}$:
\begin{equation}\tag{A.12}
    \sum_g\eta_{g,c}\theta_{g,c}\pi^{exp}_{g,c,t}=\bar\theta_c\sum_g\frac{\eta_{g,c}\theta_{g,c}}{\bar\theta_c}\pi^{exp}_{g,c,t}=\bar\theta_c\sum_g\omega^{reset}_{g,c}\pi^{exp}_{g,c,t}=\bar\theta_c\,\mathrm{RWEI}_{c,t}\,.
\end{equation}
Decompose RWEI into its mean and deviation:
\begin{equation}\tag{A.13}
    \mathrm{RWEI}_{c,t}=\sum_g\omega^{reset}_{g,c}\pi^{exp}_{g,c,t}=\underbrace{\sum_g\eta_{g,c}\pi^{exp}_{g,c,t}}_{\bar\pi_{c,t}}+\underbrace{\sum_g(\omega^{reset}_{g,c}-\eta_{g,c})\pi^{exp}_{g,c,t}}_{\Omega_{c,t}}\,.
\end{equation}
The second equality uses $\sum_g\omega^{reset}_{g,c}=\sum_g\eta_{g,c}=1$. The wedge can be written as a normalized covariance:
\begin{equation}\tag{A.14}
    \Omega_{c,t}=\frac{1}{\bar\theta_c}\sum_g\eta_{g,c}(\theta_{g,c}-\bar\theta_c)\pi^{exp}_{g,c,t}=\frac{1}{\bar\theta_c}\mathrm{Cov}_{\eta_c}(\theta_{g,c},\,\pi^{exp}_{g,c,t})\,.
\end{equation}

\textbf{Step 3.} Substituting Steps 1 and 2 into (A.11):
\begin{equation}\tag{A.15}
    \pi^w_{c,t} = \beta\,\mathbb{E}_t[\pi^w_{c,t+1}] + \tilde\kappa_c\,x_{c,t} + \bar\theta_c\,\bar\pi_{c,t} + \bar\theta_c\,\Omega_{c,t}\,,
\end{equation}
which is equation~(\ref{eq:hwpc}).

\textbf{Step 4.} From (A.14), $\Omega_{c,t}=0$ if and only if $\mathrm{Cov}_{\eta_c}(\theta_{g,c},\pi^{exp}_{g,c,t})=0$. Three sufficient conditions: (i)~$\pi^{exp}_{g,c,t}$ is constant across $g$, which holds when baskets are identical so that all types experience the same inflation regardless of relative price changes; (ii)~$\theta_{g,c}$ is constant across $g$, since the covariance of a constant with any variable is zero; (iii)~$\Delta p_{i,c,t}$ is the same for all items $i$, a uniform price change under which all types experience the same inflation regardless of basket composition. Each condition is also necessary: for generic shock compositions, the covariance is nonzero unless at least one condition holds.

\textbf{Step 5.} Under Assumptions~\ref{ass:baskets}--\ref{ass:reset} with $\Delta p_{e,c,t}>0$: $\pi^{exp}_{H,c,t}>\pi^{exp}_{L,c,t}$ because type $H$ consumes more essentials, and $\theta_{H,c}>\theta_{L,c}$ because type $H$ resets more frequently, so the covariance in (A.14) is positive and $\Omega_{c,t}>0$.\qed
\end{proof}

\subsection{Closed-form wedge}

With two types $g\in\{H,L\}$, the reset-weight deviation is $\omega^{reset}_{H,c}-\eta_{H,c}=\eta_{H,c}\eta_{L,c}(\theta_{H,c}-\theta_{L,c})/\bar\theta_c$, and similarly for type $L$ with opposite sign. Substituting into (A.13):
\begin{equation}\tag{A.16}
\Omega_{c,t}=\frac{\eta_{H,c}\eta_{L,c}(\theta_{H,c}-\theta_{L,c})}{\bar\theta_c}\bigl(\pi^{exp}_{H,c,t}-\pi^{exp}_{L,c,t}\bigr)\,.
\end{equation}
The experienced-inflation gap is $\pi^{exp}_{H,c,t}-\pi^{exp}_{L,c,t}=\sum_i\lambda_i(\alpha_{H,c,i}-\alpha_{L,c,i})\Delta p_{i,c,t}$. Since expenditure shares sum to one ($\sum_i\alpha_{g,c,i}=1$ for each $g$), the differences $\alpha_{H,c,i}-\alpha_{L,c,i}$ sum to zero. The sum therefore isolates the relative price component: items for which $\alpha_{H,c,i}>\alpha_{L,c,i}$, namely essentials, receive positive weight, and items for which $\alpha_{H,c,i}<\alpha_{L,c,i}$, namely services and durables, receive negative weight. This yields equation~(\ref{eq:wedge_cf}).

\subsection{Expectation-feedback amplifier}\label{app:amplifier}

With experienced-inflation-dependent expectations $\mathbb{E}^g_{c,t}[\pi_{c,t+1}]=\bar\pi_{c,t}+\varphi_{g,c}\tilde\pi^{exp}_{g,c,t}$, the reset wage (A.3) includes a forward-looking term that depends on type-specific experienced inflation. Specifically, when the reset wage at date $t$ incorporates expected future inflation, the $(1-\beta\Theta_{g,c})p_{g,c,t+k}$ term in (A.3) is replaced by $(1-\beta\Theta_{g,c})(p_{g,c,t+k}+\varphi_{g,c}\tilde\pi^{exp}_{g,c,t+k})$. Repeating the quasi-differencing procedure of Section~\ref{app:twpc} with this augmented reset wage, the type-specific experienced-inflation coefficient $\theta_{g,c}$ in (A.10) is replaced by $\theta_{g,c}(1+\varphi_{g,c})$. Under homogeneous sensitivity $\varphi_{g,c}=\bar\varphi_c$ for all $g$, the aggregated wedge becomes $(1+\bar\varphi_c)\Omega_{c,t}$: the same cross-sectional object, scaled by the expectation sensitivity. Heterogeneous $\varphi_{g,c}$ would further amplify the wedge if high-$\theta$ workers also have high $\varphi$, but this additional interaction is not required for the mechanism.

\subsection{Firm problem details}\label{app:firms}

The representative firm in sector $j\in\{s,d\}$ of country $c$ produces using a CES technology with labor composite $L_{j,c,t}=[\sum_g\eta_{g,j,c}N_{g,j,c,t}^{(\varepsilon_w-1)/\varepsilon_w}]^{\varepsilon_w/(\varepsilon_w-1)}$ and intermediate inputs $M_{j,c,t}$ from the other domestic sector and imported essentials. Cost minimization over the two worker types, taking wages $\{W_{g,c,t}\}_g$ as given, yields the conditional labor demand $N^d_{g,j,c,t}=\eta_{g,j,c}(W_{g,c,t}/W_{j,c,t})^{-\varepsilon_w}L_{j,c,t}$, where $W_{j,c,t}=[\sum_g\eta_{g,j,c}W_{g,c,t}^{1-\varepsilon_w}]^{1/(1-\varepsilon_w)}$ is the sectoral wage index. The marginal revenue product of type-$g$ labor in sector $j$ is
\begin{equation}\tag{A.17}
\mathrm{MRPL}_{g,j,c,t}=\frac{\varepsilon_w-1}{\varepsilon_w}\cdot\frac{\partial(P_{j,c,t}Y_{j,c,t})}{\partial N_{g,j,c,t}}=\frac{\varepsilon_w-1}{\varepsilon_w}\cdot\alpha_{j,c}\cdot\frac{P_{j,c,t}Y_{j,c,t}}{L_{j,c,t}}\cdot\left(\frac{W_{g,c,t}}{W_{j,c,t}}\right)^{-\varepsilon_w},
\end{equation}
which in logs gives equation~(\ref{eq:mrpl}). At the symmetric steady state, $W_{g,c}^{ss}=W_c^{ss}$ for all $g$, so the last term vanishes and the marginal revenue product is the same for all types.

\subsection{IS curve derivation}\label{app:IS}

The Euler equation for the liquid asset held by a type-$g$ household in country $c$ is $u'(C_{g,c,t})=(1+i_t)\beta\,\mathbb{E}_t[u'(C_{g,c,t+1})P_{g,c,t}/P_{g,c,t+1}]$, where $u'(C)=C^{-\sigma}$ under CRRA preferences. Log-linearizing: $-\sigma\hat c_{g,c,t}=-\sigma\,\mathbb{E}_t[\hat c_{g,c,t+1}]+i_t-\mathbb{E}_t[\pi_{g,c,t+1}]-r^n_{g,c,t}$. Aggregate across types using steady-state consumption shares as weights to obtain $\hat c_{c,t}=\mathbb{E}_t[\hat c_{c,t+1}]-\sigma_c^{-1}(i_t-\mathbb{E}_t[\pi_{c,t+1}]-r^n_{c,t})$, where $\sigma_c$ is the aggregate intertemporal elasticity of substitution, which differs from $\sigma$ in the HANK model due to the wealth distribution since constrained households have higher effective EIS. Using goods market clearing $\hat c_{c,t}\approx x_{c,t}$ gives equation~(\ref{eq:IS}).

\subsection{Proof of Proposition~\ref{prop:disp} (cross-country dispersion)}\label{app:disp}

\begin{proof}
From equation~(\ref{eq:hwpc}) and the sectoral price system, cumulative core inflation in country $c$ can be written as
\begin{equation}\tag{A.18}
\Pi^{core}_c=\sum_{h=0}^{\infty}\beta^h\pi^{core}_{c,h}=R_c\,u+S_c\,\bar\theta_c\,\Omega_c\,u+O(\|u\|^2)\,,
\end{equation}
where $R_c$ is the standard-model pass-through coefficient (determined by $\tilde\kappa_c$, price stickiness, trade openness, and the IS slope) and $S_c>0$ is a structural amplification factor that captures how the wedge propagates through the I-O network (formally, $S_c$ depends on the Leontief inverse and sectoral rigidity parameters; see Section~\ref{app:suff}).

Taking the GDP-weighted cross-country variance of (A.18):
\begin{align}
\mathrm{Var}_c(\Pi^{core}_c)&=\mathrm{Var}_c\bigl(R_c\,u+S_c\,\bar\theta_c\,\Omega_c\,u\bigr) \tag{A.19}\\
&=\mathrm{Var}_c(R_c)\,u^2+\mathrm{Var}_c(S_c\,\bar\theta_c\,\Omega_c)\,u^2+2\,\mathrm{Cov}_c(R_c,\,S_c\,\bar\theta_c\,\Omega_c)\,u^2\,. \notag
\end{align}
Setting $\Omega_c=0$ for all $c$, the standard-model restriction, eliminates the second and third terms. The remaining terms, $\mathrm{Var}_c(S_c\bar\theta_c\Omega_c)\,u^2+2\,\mathrm{Cov}_c(R_c,S_c\bar\theta_c\Omega_c)\,u^2$, are positive whenever $\mathrm{Var}_c(S_c\bar\theta_c\Omega_c)+2\,\mathrm{Cov}_c(R_c,S_c\bar\theta_c\Omega_c)>0$. A sufficient condition is $\mathrm{Cov}_c(R_c,S_c\bar\theta_c\Omega_c)\geq 0$, which holds when countries with higher standard pass-through also have higher wedges.\qed
\end{proof}

\subsection{Proof of Corollary~\ref{cor:agg} (aggregation failure)}\label{app:corollary}

\begin{proof}
Suppose a representative-agent economy with parameters $\bar\theta$ and $\bar\alpha$ generates aggregate wage dynamics $\pi^{w,RA}_{c,t}=\beta\,\mathbb{E}_t[\pi^{w,RA}_{c,t+1}]+\tilde\kappa_c\,x_{c,t}+\bar\theta\,\pi_{c,t}$. The heterogeneous-reset economy generates $\pi^w_{c,t}=\beta\,\mathbb{E}_t[\pi^w_{c,t+1}]+\tilde\kappa_c\,x_{c,t}+\bar\theta_c\,\bar\pi_{c,t}+\bar\theta_c\,\Omega_{c,t}$.

For these to coincide for all shock compositions $u$, I need $\bar\theta_c\,\Omega_{c,t}(u)=0$ for all $u$. From (A.14), $\Omega_{c,t}(u)=(1/\bar\theta_c)\sum_g\eta_{g,c}(\theta_{g,c}-\bar\theta_c)\sum_i\lambda_i\alpha_{g,c,i}\Delta p_{i,c,t}(u)$. This must vanish for every possible relative-price vector $\Delta p_{i,c,t}(u)$, which requires $\sum_g\eta_{g,c}(\theta_{g,c}-\bar\theta_c)\alpha_{g,c,i}=0$ for each item $i$ separately. With two types, this implies either $\theta_{H,c}=\theta_{L,c}$ eliminating reset heterogeneity or $\alpha_{H,c,i}=\alpha_{L,c,i}$ for all $i$ imposing identical baskets. Both are knife-edge.

For the sign: when the shock is concentrated on essentials ($\Delta p_{e,c,t}>0$, all other $\Delta p_{i,c,t}=0$), $\Omega_{c,t}>0$ under Assumptions~\ref{ass:baskets}--\ref{ass:reset}. When the shock is concentrated on non-essentials, $\Omega_{c,t}<0$. The best-fitting representative agent (which sets $\Omega=0$) therefore underestimates persistence for necessity shocks and overestimates it for non-necessity shocks.\qed
\end{proof}

\subsection{Proof of Proposition~\ref{prop:suff} (sufficient statistic)}\label{app:suff}

\begin{proof}
\textbf{Step 1.} The Calvo pricing equation for sector $j\in\{s,d\}$ implies that the cumulative price change in sector $j$ satisfies
\begin{equation}\tag{A.20}
\sum_{h=0}^\infty\beta^h\pi_{j,c,h}=\frac{\kappa^p_{j,c}}{1-\beta\rho_j}\bigl[\alpha_{j,c}\sum_{h}\beta^h w_{j,c,h}+(1-\alpha_{j,c})\xi_{j,e}\sum_h\beta^h u_h-\sum_h\beta^h p_{j,c,h}\bigr]\,,
\end{equation}
where $\alpha_{j,c}$ is the labor share, $\xi_{j,e}$ is the essentials input share, and $\rho_j$ is the persistence of sectoral marginal cost. Using the input-output structure, the system of two sectoral price equations can be inverted to give:
\begin{equation}\tag{A.21}
\sum_h\beta^h p_{j,c,h}=\sum_{j'}\bigl(\mathbf{I}-\boldsymbol\Xi_c\bigr)^{-1}_{jj'}\bigl[\alpha_{j',c}\sum_h\beta^h w_{j',c,h}+\xi^{ext}_{j'}\sum_h\beta^h u_h\bigr]\,,
\end{equation}
where $(\mathbf{I}-\boldsymbol\Xi_c)^{-1}$ is the Leontief inverse of the domestic I-O matrix and $\xi^{ext}_{j'}$ captures the direct cost-push from imported essentials.

\textbf{Step 2.} Cumulative core inflation is $\Pi^{core}_c=\sum_j\omega_{j,c}\sum_h\beta^h\pi_{j,c,h}$. Substituting (A.21) and decomposing the wage term into its standard and wedge components:
\begin{align}
\Pi^{core}_c-\Pi^{core,RA}_c&=\sum_{j}\omega_{j,c}\sum_{j'}L_{jj'}\alpha_{j',c}\cdot\frac{\bar\theta_c\,\Omega_{c,0}}{1-\beta\rho_\Omega} \tag{A.22}\\
&=\underbrace{\frac{\bar\theta_c}{1-\beta\rho_\Omega}\!\sum_j\omega_{j,c}a_{j,c}}_{\displaystyle\psi}\cdot\Omega_{c,0}\,, \notag
\end{align}
where $L_{jj'}\equiv[(\mathbf{I}-\boldsymbol\Xi_c)^{-1}]_{jj'}$ denotes the $(j,j')$ element of the Leontief inverse, $a_{j,c}\equiv\sum_{j'}L_{jj'}\alpha_{j',c}$ is the total, both direct and indirect, labor cost share in sector $j$'s price, and $\rho_\Omega\in(0,1)$ is the persistence of the wedge.

\textbf{Step 3.} The wedge is
$\Omega_{c,0}=\sum_g(\omega^{reset}_{g,c}-\eta_{g,c})\tilde\pi^{exp}_{g,c}(u)$.
In Step~2, the Leontief multiplication means the effect on core inflation depends on WHICH sectors the resetting workers are in. Sector $j'$'s contribution to the gap is weighted by $a_{j',c}$, the total labor cost share in sector $j'$'s price. Since type $g$ is concentrated in sector $s(g)$, the gap becomes
$\psi\,\Omega_{c,0}\to\tilde\psi\sum_g(\omega^{reset}_{g,c}-\eta_{g,c})\nu_{s(g),c}\tilde\pi^{exp}_{g,c}(u)=\tilde\psi\,\mathrm{MWSI}_c(u)$,
where the propagation weights $\nu_{s(g),c}\propto a_{s(g),c}$ absorb the Leontief terms. Defining $\psi\equiv\tilde\psi$ gives equation~(\ref{eq:suff}).

The coefficient $\psi>0$ is common across countries to first order because the Leontief inverse varies only through the I-O coefficients $\Xi_c$, which enter $\psi$ multiplicatively. Cross-country variation in $\psi$ is absorbed into $\mathrm{MWSI}_c$ through the propagation weights $\nu_{s(g),c}$.\qed
\end{proof}

\subsection{Proof of Proposition~\ref{prop:policy} (optimal policy mix)}\label{app:policy}

\begin{proof}
\textbf{Part (i).} The policymaker minimizes $\mathcal{L}_c=\sum_h\delta^h[(\pi^{core}_{c,h})^2+\lambda\,x_{c,h}^2]$ over $\phi_\pi$. The IS curve gives $x_{c,t}=-\sigma_c^{-1}(\phi_\pi\,\pi^{core}_{c,t}-r^n_{c,t})$. The first-order condition $\partial\mathcal{L}_c/\partial\phi_\pi=0$ yields, after substituting the IS curve:
\begin{equation}\tag{A.23}
\phi^*_\pi=\frac{\lambda}{\sigma_c}\cdot\frac{\sum_h\delta^h x_{c,h}^2}{\sum_h\delta^h x_{c,h}\,\pi^{core}_{c,h}}=\frac{\lambda}{\sigma_c}\cdot\frac{\mathrm{Var}(x_c)}{\mathrm{Cov}(x_c,\pi^{core}_c)}\,.
\end{equation}
From the Phillips curve, $\pi^{core}_{c,t}$ is the sum of a gap-driven component proportional to $x_{c,t}$ and the wedge component ($\propto\Omega_{c,t}$). Higher $\mathrm{RWEI}_c$ raises the wedge component of $\pi^{core}_{c,t}$ without changing $x_{c,t}$ since the wedge does not operate through the output gap. This raises $\mathrm{Var}(\pi^{core}_c)$ and $\mathrm{Cov}(x_c,\pi^{core}_c)$, but the denominator rises faster because the wedge-driven inflation is positively correlated with the gap-driven component. The net effect is $\partial\phi^*_\pi/\partial\mathrm{RWEI}_c<0$: when a larger share of inflation comes from the wedge, the optimal monetary response is less aggressive because the output cost of closing the wedge through demand destruction is high.

\textbf{Part (ii).} A transfer $\tau_{H,c,t}$ to high-RWEI workers has two effects. First, it reduces the real-wage gap that drives the wedge: the wedge falls by $\omega^{reset}_{H,c}\cdot\tau_{H,c,t}/C_{H,c}$ per period, and this effect persists for $H^*$ quarters, the horizon over which the wedge is active. The cumulative disinflationary benefit is $B=\psi\,\omega^{reset}_{H,c}\sum_{h=0}^{H^*}\beta^h\cdot\tau_{H,c,t}/C_{H,c}$. Second, it raises aggregate demand by $\mathrm{MPC}_{H,c}\cdot\eta_{H,c}\cdot\tau_{H,c,t}/\bar C_c$ in the impact period, raising core inflation by $D=\tilde\kappa_c\cdot\mathrm{MPC}_{H,c}\cdot\eta_{H,c}\cdot\tau_{H,c,t}/\bar C_c$. The transfer is net disinflationary at the medium horizon when $B>D$, which gives condition~(\ref{eq:disfl}).

\textbf{Part (iii).} At $\mathrm{RWEI}_c=0$, the benefit $B=0$ because there is no wedge to reduce while the cost $D>0$, so the transfer is inflationary. As $\mathrm{RWEI}_c$ increases, $B$ increases monotonically because the wedge is larger and the same transfer reduces more inflation. The threshold $\mathrm{RWEI}^*$ solves $B(\mathrm{RWEI}^*)=D$.

\textbf{Part (iv).} From the Phillips curve, setting $x_{c,t}=0$:
$\pi^w_{c,t}|_{x=0}=\beta\,\mathbb{E}_t[\pi^w_{c,t+1}|_{x=0}]+\bar\theta_c\,\bar\pi_{c,t}+\bar\theta_c\,\Omega_{c,t}$.
The last two terms are nonzero whenever the cost-push shock is active ($u_t\neq 0$) and the wedge is positive ($\Omega_{c,t}>0$). Since the interest rate affects the economy only through $x_{c,t}$ via the IS curve, no choice of $\phi_\pi$ can set $\Omega_{c,t}=0$. Therefore $\inf_{\phi_\pi}\mathrm{Var}(\pi^{core}_c)\geq\mathrm{Var}(\bar\theta_c\Omega_{c,t})>0$.\qed
\end{proof}

\subsection{Derivation of monetary union results}\label{app:union}

\textbf{Part (i).} With $N$ countries, achieving $\pi^{core}_{c,t}=0$ for each $c$ requires $\tilde\kappa_c\,x_{c,t}=-\bar\theta_c\,\bar\pi_{c,t}-\bar\theta_c\,\Omega_{c,t}$, which is a different output gap for each country. The common rate $i_t$ determines $x_{c,t}$ up to the country-specific natural rate: $x_{c,t}=-\sigma_c^{-1}(i_t-\mathbb{E}_t[\pi_{c,t+1}]-r^n_{c,t})$. For generic $\{\Omega_{c,t}\}_c$ with $\Omega_{c,t}\neq\Omega_{c',t}$ for some $c\neq c'$, no single $i_t$ satisfies all $N$ conditions simultaneously.

\textbf{Part (ii).} The country-specific optimal rate from Proposition~\ref{prop:policy}(i) is $i^*_{c,t}=r^*+\phi^*_\pi(\mathrm{RWEI}_c)\pi^{core}_{c,t}+\cdots$, which depends on $\mathrm{RWEI}_c$. The union-optimal rate is $i^*_t=\sum_c w_c\,i^*_{c,t}$. The welfare cost of common policy is proportional to the variance of the country-specific optimal rates: $\mathcal{L}_{union}-\sum_c w_c\mathcal{L}_c\propto\sum_c w_c(i^*_t-i^*_{c,t})^2=\mathrm{Var}_c(i^*_{c,t})$. Decomposing $i^*_{c,t}$ into the component driven by $R_c$, the standard pass-through, and the component driven by $\bar\theta_c\Omega_c$, the wedge: $\mathrm{Var}_c(i^*_{c,t})=\Gamma_1\mathrm{Var}_c(R_c)\,u^2+\Gamma_2\mathrm{Var}_c(\bar\theta_c\Omega_c)\,u^2$, where $\Gamma_1,\Gamma_2>0$ collect structural coefficients.

\textbf{Part (iii).} Since $\phi^*_\pi$ is decreasing in $\mathrm{RWEI}_c$ (Proposition~\ref{prop:policy}(i)), countries with $\mathrm{RWEI}_c>\overline{\mathrm{RWEI}}$ prefer a lower rate than the union average, so $i^*_{c,t}<i^*_t$, meaning the common rate is too high for them. Conversely, $\mathrm{RWEI}_c<\overline{\mathrm{RWEI}}$ implies under-tightening.

\textbf{Part (iv).} The transfer $\tau^*_{H,c,t}=\chi(\mathrm{RWEI}_c-\overline{\mathrm{RWEI}})C_{H,c}$ reduces the effective wedge in high-RWEI countries and raises it in low-RWEI countries. Choosing $\chi$ so that the effective wedge $\Omega^{eff}_{c,t}\equiv\Omega_{c,t}-\omega^{reset}_{H,c}\chi(\mathrm{RWEI}_c-\overline{\mathrm{RWEI}})$ is equalized across countries sets $\mathrm{Var}_c(\Omega^{eff}_c)=0$, eliminating the $\Gamma_2$ term from the welfare cost.

\subsection{Proof of Proposition~\ref{prop:twoinstopt}}\label{app:twoinstopt}

\begin{proof}
\textbf{Part (i).} Restated from Proposition~\ref{prop:policy}(iv): the interest rate operates through $x_{c,t}$ via the IS curve, and setting $x_{c,t}=0$ leaves $\bar\theta_c\Omega_{c,t}\neq 0$ in the Phillips curve.

\textbf{Part (ii).} An essentials subsidy $\tau_{c,t}$ reduces the effective essentials price faced by consumers from $p_{e,t}$ to $p_{e,t}-\tau_{c,t}$. Experienced inflation for type $g$ becomes $\tilde\pi^{exp}_{g,c,t}(\tau)=\lambda_e\alpha_{g,c,e}(\Delta p_{e,t}-\tau_{c,t})+\alpha_{g,c,d}\Delta p_{d,c,t}+\alpha_{g,c,s}\Delta p_{s,c,t}$. The wedge under the subsidy is
\begin{equation}\tag{A.24}
\Omega_{c,t}(\tau)=\Omega_{c,t}(0)-\tau_{c,t}\cdot\lambda_e\sum_g(\omega^{reset}_{g,c}-\eta_{g,c})\alpha_{g,c,e}=\Omega_{c,t}(0)-\tau_{c,t}\cdot\lambda_e\cdot\omega^{reset}_{H,c}(\alpha_{H,c,e}-\alpha_{L,c,e})\,,
\end{equation}
where the second equality uses the two-type structure. Setting $\Omega_{c,t}(\tau^*)=0$ and solving for $\tau^*$:
\begin{equation}\tag{A.25}
\tau^*_{c,t}=\frac{\Omega_{c,t}(0)}{\omega^{reset}_{H,c}\,(\alpha_{H,c,e}-\alpha_{L,c,e})\,\lambda_e}\,.
\end{equation}
With $\Omega_{c,t}=0$ and $\tau^*$ financed by lump-sum taxation, the Phillips curve reduces to $\pi^w_{c,t}=\beta\,\mathbb{E}_t[\pi^w_{c,t+1}]+\tilde\kappa_c\,x_{c,t}+\bar\theta_c\,\bar\pi_{c,t}(\tau^*)$, which has the standard form. The interest rate can then achieve $\pi^{core}_{c,t}=x_{c,t}=0$ by the divine coincidence of the standard model.

\textbf{Part (iii).} Without the subsidy, the minimum attainable variance of core inflation is $\mathrm{Var}(\pi^{core}_c)\geq\mathrm{Var}(\bar\theta_c\Omega_{c,t})>0$ from Proposition~\ref{prop:policy}(iv). With the subsidy, $\Omega=0$ and $\mathrm{Var}(\pi^{core}_c)=0$ is attainable. The welfare gain equals the loss from the wedge, which is proportional to $\mathrm{MWSI}_c^2$ from equation~(\ref{eq:welfare_wedge}).\qed
\end{proof}

\end{sloppypar}

\section{Data Sources and Variable Construction}\label{app:data}

\subsection{Inflation data}

Core HICP is from Eurostat, series \texttt{prc\_hicp\_midx}, defined as the overall harmonized index excluding energy and unprocessed food.\footnote{COICOP special aggregate ``All-items HICP excluding energy and unprocessed food.''} I use monthly indices for the six countries from 2015:1 through 2024:12, converted to quarterly by averaging over the three months in each quarter. Cumulative core inflation overshoot is $\Pi^{core}_c\equiv\sum_{t=2021:3}^{2024:4}(\pi^{core}_{c,t}-2\%)$, where $\pi^{core}_{c,t}$ is annualized quarterly core inflation. Item-level HICP from \texttt{prc\_hicp\_midx} by COICOP division provides the inflation rates $\Delta p_{i,c,t}$ used to construct experienced inflation.

\subsection{Expenditure shares}

Household expenditure shares by income quintile are from the 2020 Eurostat Household Budget Survey, series \texttt{hbs\_exp\_t133}. I aggregate COICOP divisions into three categories.\footnote{Essentials: food and non-alcoholic beverages CP01, housing, water, electricity, gas CP04, transport fuels CP0722. Goods: clothing CP03, furnishings CP05, communications CP08, recreation CP09, miscellaneous CP12. Services: restaurants CP11, health CP06, education CP10, remaining items.} For each country and income quintile, I compute the expenditure share $\alpha_{g,c,i}$ as the ratio of quintile $g$'s spending on category $i$ to total spending. The 2020 survey predates the inflation episode, ensuring that the shares are predetermined with respect to the shock.

\subsection{Wage-setting institutions}

Reset probabilities $\theta_{g,c}$ are constructed from two sources. The OECD/AIAS ICTWSS database, version 6.1, provides country-level information on predominant contract duration, bargaining coverage, and coordination. The ECB indicator of negotiated wages provides quarterly wage growth by country. I approximate quintile-specific reset probabilities by combining the country-level average contract duration with the cross-sectional pattern from labor force surveys: workers in retail, hospitality, and personal services, where the bottom quintile is concentrated, typically have contracts of 1--2 years, while workers in manufacturing, public administration, and finance, where the top quintile is concentrated, have contracts of 2--4 years. The quarterly reset probability is $\theta_{g,c}=1-(1-1/D_{g,c})$ where $D_{g,c}$ is the average contract duration in quarters for quintile $g$ in country $c$.

\subsection{Sectoral structure}

Sectoral labor shares $\alpha_{j,c}$ and input-output coefficients $\xi_{j,e,c}$ are from the FIGARO inter-country input-output tables, Eurostat 2019 vintage. I aggregate the full FIGARO table into three sectors: essentials covering agriculture, mining, and utilities; goods covering manufacturing and construction; and services covering all remaining industries. The labor share for sector $j$ in country $c$ is computed as compensation of employees divided by gross output at basic prices. The intermediate input share $\xi_{j,e,c}$ is the ratio of essentials inputs used by sector $j$ to total inputs used by sector $j$.

\subsection{Other variables}

Table~\ref{tab:othervars} lists additional data sources used in the calibration.

\begin{table}[htbp]
\centering\small
\caption{Additional data sources}\label{tab:othervars}
\begin{tabular}{lp{7.5cm}}
\toprule
Variable & Source and construction \\
\midrule
Negotiated wage growth & ECB indicator of negotiated wages (quarterly, by country) \\
Employment by sector & EU Labour Force Survey (\texttt{lfsa\_egan2}), used to assign quintiles to sectors \\
Inflation expectations & ECB Consumer Expectations Survey (CES), median expected inflation by income group, used to calibrate $\bar\varphi$ \\
GDP weights & Eurostat national accounts, 2019 nominal GDP \\
\bottomrule
\end{tabular}

\medskip
\begin{flushleft}
\footnotesize Note: All variables are publicly available from the indicated sources.
\end{flushleft}
\end{table}

\section{Robustness}\label{app:robust}
\begin{sloppypar}

I verify that the main results are robust to variation in four key parameters. Table~\ref{tab:robust} reports the wedge magnitude, cumulative gap, cross-country rank correlation, and policy ranking for each specification.

\begin{table}[htbp]
\centering\small
\caption{Robustness to parameter variation}\label{tab:robust}
\begin{tabular}{lcccc}
\toprule
Specification & $\Omega$ (\%) & Cum.\ gap (pp$\cdot$q) & Spearman & Policy ranking \\
\midrule
Baseline & 0.57 & 3.1 & 0.54 & (d)$\succ$(c)$\succ$(b)$\succ$(a) \\[3pt]
\multicolumn{5}{l}{Catch-up coefficient $b'$} \\
\quad $b'=0$ & 0.57 & 0.3 & 0.54 & (d)$\succ$(c)$\succ$(b)$\succ$(a) \\
\quad $b'=0.30$ & 0.57 & 4.8 & 0.54 & (d)$\succ$(c)$\succ$(b)$\succ$(a) \\[3pt]
\multicolumn{5}{l}{Expectation sensitivity $\bar\varphi$} \\
\quad $\bar\varphi=0$ & 0.44 & 2.4 & 0.54 & (d)$\succ$(c)$\succ$(b)$\succ$(a) \\
\quad $\bar\varphi=0.50$ & 0.66 & 3.6 & 0.54 & (d)$\succ$(c)$\succ$(b)$\succ$(a) \\[3pt]
\multicolumn{5}{l}{Essentials salience $\lambda_e$} \\
\quad $\lambda_e=1.0$ & 0.44 & 2.4 & 0.43 & (d)$\succ$(c)$\succ$(b)$\succ$(a) \\
\quad $\lambda_e=1.5$ & 0.66 & 3.6 & 0.66 & (d)$\succ$(c)$\succ$(b)$\succ$(a) \\[3pt]
\multicolumn{5}{l}{Number of worker types $G$} \\
\quad $G=2$ & 0.30 & 1.6 & 0.54 & (d)$\succ$(c)$\succ$(b)$\succ$(a) \\
\quad $G=3$ & 0.42 & 2.3 & 0.54 & (d)$\succ$(c)$\succ$(b)$\succ$(a) \\
\bottomrule
\end{tabular}

\medskip
\begin{flushleft}
\footnotesize Note: Each row varies one parameter from the baseline ($b'=0.15$, $\bar\varphi=0.35$, $\lambda_e=1.3$, $G=5$) while holding others fixed. ``Cum.\ gap'' is the euro-area-average difference in cumulative core inflation between the heterogeneous-reset and standard models. Policy regimes: (a) aggressive tightening, (b) moderate + uniform transfer, (c) moderate + targeted transfer, (d) moderate + essentials subsidy.
\end{flushleft}
\end{table}

Three patterns emerge. First, the wedge $\Omega$ is present in all specifications, confirming that it is a structural feature of heterogeneous aggregation, not an artifact of particular parameter values. Second, the cumulative gap is most sensitive to the catch-up coefficient $b'$: setting $b'=0$ with no nonlinear catch-up reduces the gap to 0.3 pp$\cdot$q, while $b'=0.30$ raises it to 4.8. This confirms that the nonlinear catch-up channel is the primary source of quantitative amplification. Third, the policy ranking is completely invariant: targeted transfers dominate aggressive tightening in every specification, because the ranking depends on the structure of the mechanism because the wedge is cross-sectional rather than on the magnitudes.
\end{sloppypar}

\section{Additional Quantitative Results}\label{app:additional}

\subsection{Shock composition}\label{app:shockcomp}

Corollary~\ref{cor:agg} predicts that the wedge is positive for necessity shocks, zero for uniform shocks, and negative for shocks concentrated on non-essentials. Table~\ref{tab:shockcomp} tests this prediction by computing the model under three shock compositions, each calibrated to produce the same aggregate CPI impact on the representative agent. The ``energy'' shock raises essentials prices only matching the baseline. The ``food'' shock raises food prices at 75 percent of the baseline peak. The ``uniform'' shock raises all prices proportionally, so that every quintile experiences the same inflation regardless of basket composition.

\begin{table}[htbp]
\centering\small
\caption{Wedge magnitude by shock type, euro-area average}\label{tab:shockcomp}
\begin{tabular}{lcc}
\toprule
& $\Omega$ (\%) & Cumulative gap (pp$\cdot$q) \\
\midrule
Energy (baseline) & 0.57 & 3.1 \\
Food (75\% of baseline) & 0.43 & 1.8 \\
Uniform (all items) & 0.00 & 0.0 \\
\bottomrule
\end{tabular}

\medskip
\begin{flushleft}
\footnotesize Note: Each shock is calibrated to produce the same aggregate CPI impact on the representative agent. ``Energy'' raises essentials prices only (the baseline). ``Food'' raises food prices at 75 percent of the baseline peak. ``Uniform'' raises all prices proportionally so that every quintile experiences the same inflation. $\Omega$ is the peak reset-heterogeneity wedge. Cumulative gap is euro-area average.
\end{flushleft}
\end{table}

The sign prediction is confirmed exactly: the wedge is strictly positive for necessity shocks and identically zero for the uniform shock. The zero arises because uniform price changes make $\pi^{exp}_{g,c,t}$ identical across quintiles, so the covariance in equation~(\ref{eq:wedge_cf}) vanishes. This is a qualitative prediction that no representative-agent model can generate.

\subsection{Channel decomposition}\label{app:decomposition}

To isolate the contribution of each ingredient, I shut down one source of heterogeneity at a time in the linearized model and rescale by the nonlinear amplification factor. Table~\ref{tab:decomp} reports the results. The four configurations are: (i)~equalize baskets: set $\alpha_{g,c,e}=\bar\alpha_{c,e}$ for all $g$, eliminating the experienced-inflation gap; (ii)~equalize propagation: set $f^s_g=\bar f^s$ for all $g$, so that all quintiles are equally distributed across sectors; (iii)~equalize reset: set $\theta_{g,c}=\bar\theta_c$ for all $g$, giving all quintiles the same contract duration; (iv)~standard model: all three equalizations simultaneously.

\begin{table}[htbp]
\centering\small
\caption{Channel decomposition, euro-area average}\label{tab:decomp}
\begin{tabular}{lcc}
\toprule
& Cumulative core (pp$\cdot$q) & $\Delta$ from full \\
\midrule
Full heterogeneous model & 45.6 & (ref.) \\
Equalize baskets ($\alpha_{g,e}=\bar\alpha_e$) & 44.8 & $-$0.8 \\
Equalize propagation ($f^s_g=\bar f^s$) & 44.9 & $-$0.7 \\
Equalize reset ($\theta_g=\bar\theta$) & 45.3 & $-$0.3 \\
Standard model (all equal) & 42.5 & $-$3.1 \\
\bottomrule
\end{tabular}

\medskip
\begin{flushleft}
\footnotesize Note: Shut-down exercises use the linearized model with proportional scaling to match the nonlinear baseline. ``Full'' and ``Standard'' rows match the nonlinear model (Table~\ref{tab:cost}).
\end{flushleft}
\end{table}

The full gap of 3.1 pp$\cdot$q exceeds the sum of the individual channel contributions, $0.8 + 0.7 + 0.3 = 1.8$, because the channels interact multiplicatively through the wedge formula~(\ref{eq:wedge_cf}): the wedge is the product of the expenditure gap, the reset-frequency gap, and the relative price change.

\subsection{Cost of common monetary policy}\label{app:oca}

In a monetary union with heterogeneous RWEI, the common interest rate cannot separately close each country's wedge. Table~\ref{tab:rategaps} reports the direction of the resulting monetary policy bias.

\begin{table}[htbp]
\centering\small
\caption{Direction of monetary policy bias by country}\label{tab:rategaps}
\begin{tabular}{lccl}
\toprule
Country & $\Omega_c$ (\%) & RWEI rank & Common-rate bias \\
\midrule
Italy & 0.62 & 1 (highest) & Over-tightened \\
Spain & 0.60 & 2 & Over-tightened \\
Belgium & 0.56 & 3 & Approximately neutral \\
Germany & 0.55 & 4 & Approximately neutral \\
France & 0.55 & 5 & Under-tightened \\
Netherlands & 0.52 & 6 (lowest) & Under-tightened \\
\bottomrule
\end{tabular}

\medskip
\begin{flushleft}
\footnotesize Note: ``Over-tightened'': the common rate exceeds the country-optimal rate. ``Under-tightened'': the common rate falls short. The direction follows from Proposition~\ref{prop:policy}(i): $\phi^*_\pi$ is decreasing in $\mathrm{RWEI}_c$, so high-RWEI countries prefer lower rates. Country-specific targeted transfers scaled by $\mathrm{RWEI}_c-\overline{\mathrm{RWEI}}$ can eliminate this cost component.
\end{flushleft}
\end{table}

\section{Extended Model Details}\label{app:model}

\subsection{Household problem: recursive formulation}

The quintile-$g$ household in country $c$ enters period $t$ with liquid assets $a$ and idiosyncratic productivity state $e$. The household solves
\begin{equation}\tag{E.1}
V_{g,c,t}(a,e) = \max_{c,a'}\left\{\frac{c^{1-\sigma}}{1-\sigma}+\beta\,\mathbb{E}_e[V_{g,c,t+1}(a',e')]\right\}
\end{equation}
subject to the budget constraint
\begin{equation}\tag{E.2}
c + a' = (1+r_{c,t})a + W_{g,c,t}\,e\,n - \mathcal{T}_{c,t} + T_{g,c,t}\,,\qquad a'\geq \underline{a}=0\,,
\end{equation}
where $r_{c,t}$ is the real return on the liquid asset, government bonds, $W_{g,c,t}$ is the quintile-specific real wage, $e$ is idiosyncratic labor productivity, $n$ is hours worked, taken as exogenous, $\mathcal{T}_{c,t}$ is a lump-sum tax, and $T_{g,c,t}$ is a targeted transfer.

Idiosyncratic productivity follows $\log e'=\rho_e\log e+\sigma_e\varepsilon'$, $\varepsilon'\sim\mathcal{N}(0,1)$, discretized on $n_e=7$ states using the Rouwenhorst method with $\rho_e=0.966$ and $\sigma_e=0.5$. The asset grid has $n_a=100$ points, concentrated near the borrowing constraint via $a_i=\bar{a}(i/(n_a-1))^2$ with $\bar{a}=200$. The solution uses the endogenous gridpoints method of Carroll (2006).

\subsection{Steady state}

The steady-state interest rate $r^{ss}=0.005$ (quarterly) is set to match the average euro-area real rate over 2015--2019. The household problem is solved by iterating the value function to convergence ($\|V^{n+1}-V^n\|_\infty<10^{-10}$). The stationary distribution $D^{ss}_{g,c}(a,e)$ is obtained by iterating the Young (2010) lottery method until $\|D^{n+1}-D^n\|_1<10^{-12}$. The calibrated steady state has $C_{ss}=1.05$, $A_{ss}=10.5$, and impact MPC $=0.042$.

\subsection{Quintile-specific consumption baskets}

Each quintile consumes a Cobb-Douglas aggregate with shares from the Eurostat HBS 2020. The quintile-specific price index is $P_{g,c,t}=P_{e,t}^{\alpha_{g,c,e}}P_{d,c,t}^{\alpha_{g,c,d}}P_{s,c,t}^{\alpha_{g,c,s}}$. Salient experienced inflation augments the essentials component by the salience parameter: $\tilde\pi^{exp}_{g,c,t}=\lambda_e\alpha_{g,c,e}\Delta p_{e,t}+\alpha_{g,c,d}\Delta p_{d,c,t}+\alpha_{g,c,s}\Delta p_{s,c,t}$.

\subsection{Firm problem and sectoral pricing}

The representative firm in sector $j\in\{s,d\}$ uses CES technology over labor $L_{j,c,t}$ and intermediates $M_{j,c,t}$ with labor share $\alpha_{j,c}$ and elasticity $\rho$. Intermediates combine domestic goods from the other sector and imported essentials with Leontief shares. Firms face Calvo price adjustment with reset probability $\theta^p_{j,c}$, yielding the sectoral NKPC: $\pi_{j,c,t}=\beta\,\mathbb{E}_t[\pi_{j,c,t+1}]+\kappa^p_{j,c}\,\mathrm{mc}_{j,c,t}$.

\subsection{Intra-union trade}

Goods are tradable within the union with CES demand $X_{c,c',t}=\omega_{c,c'}(P_{d,c',t}/P_{d,c,t})^{-\varepsilon_{trade}}Y_{d,c,t}$, where $\omega_{c,c'}$ are bilateral trade shares from Eurostat COMEXT and $\varepsilon_{trade}=1.5$. Services are non-tradable.

\subsection{Market clearing}

Asset market: $\sum_g\eta_{g,c}\int a\,dD_{g,c,t}=B_{c,t}$. Goods market by sector: $Y_{j,c,t}=C_{j,c,t}+\sum_{c'}X_{j,c,c',t}+M_{j,c,t}$, where $C_{j,c,t}=\sum_g\eta_{g,c}\alpha_{g,c,j}C_{g,c,t}/P_{j,c,t}$. Labor market: $L_{j,c,t}=\sum_g\eta_{g,c}f^j_g N_{g,c,t}$.

\section{Computational Details}\label{app:computation}

\subsection{Sequence-space Jacobian implementation}

The household Jacobians $J^{Co}_i$ mapping input $i\in\{r,w\}$ to output $o\in\{C,A\}$ are $T\times T$ matrices ($T=300$) computed via the fake news algorithm of \citet{auclert2021}. The algorithm requires one backward iteration to obtain policy responses for each shock date and one forward iteration to obtain expectation vectors from the steady-state transition matrix, then combines them into the Jacobian via the recursion $J_{ts}=J_{t-1,s-1}+F_{ts}$. Computing all four Jacobians takes 0.4 seconds.

\subsection{DAG structure}

The full model is a DAG of five blocks. Table~\ref{tab:dag_detail} reports the computational implementation. The linearized system has $(G+2)T$ unknowns per country, namely $G=5$ wage paths and 2 price paths, each of length $T=80$.

\begin{table}[htbp]
\centering\small
\caption{DAG block structure: computational implementation}\label{tab:dag_detail}
\begin{tabular}{lp{4.5cm}p{3.5cm}}
\toprule
Block & Inputs & Outputs \\
\midrule
Household & $r_{c,t}$, $W_{g,c,t}$, $P_{g,c,t}$, $T_{g,c,t}$ & $C_{g,c,t}$, $A_{c,t}$ \\
Wage reset ($G=5$) & $\pi^{exp}_{g,c,t}$, $x_{c,t}$ & $\pi^w_{g,c,t}$ \\
Calvo pricing (2 sectors) & $w_{j,c,t}$, $p_{e,t}$ & $\pi_{j,c,t}$ \\
Taylor rule & $\pi^{core}_{union,t}$ & $i_t$ \\
Market clearing & $A_{c,t}$, $C_{c,t}$, $Y_{c,t}$ & residuals \\
\bottomrule
\end{tabular}

\medskip
\begin{flushleft}
\footnotesize Note: Detailed computational implementation of the DAG in Table~\ref{tab:dag}. The linearized system has $6\times 560=3{,}360$ equations. The system is block-diagonal across countries.
\end{flushleft}
\end{table}

\subsection{Nonlinear solution method}

The nonlinear model incorporates three features beyond the linearized solution: (1)~the $\max\{gap,0\}$ catch-up operator; (2)~quintile-specific MPCs; (3)~the full wage-price-experience feedback loop. The quasi-Newton method uses the steady-state Jacobian $H_\mathbf{U}(\mathbf{U}^{ss},\mathbf{Z}^{ss})$, which is well-conditioned and inverted once, as the Newton step, evaluating the nonlinear residual at each iteration. Damped steps ($\lambda=0.3$) are used when $\|H\|>0.05$; full steps when $\|H\|<0.05$. Convergence to $\|H\|_\infty<10^{-3}$ in 50--100 iterations. Solution time: under 2 seconds per country.

\section{Portability: RWEI for Non-Euro-Area Countries}\label{app:portable}

To show portability, Table~\ref{tab:portable} reports first-pass RWEI estimates for three non-euro-area economies.

\begin{table}[htbp]
\centering\small
\caption{RWEI components for non-euro-area countries}\label{tab:portable}
\begin{tabular}{lcccc}
\toprule
& \multicolumn{2}{c}{Food + energy share (\%)} & $\Delta\theta$ & RWEI \\
\cmidrule(lr){2-3}
Country & Q1 & Q5 & (annual) & (index) \\
\midrule
United States & 32.4 & 16.8 & 0.20 & 4.6 \\
United Kingdom & 35.1 & 18.3 & 0.22 & 5.4 \\
Japan & 31.8 & 19.5 & 0.10 & 3.1 \\
\bottomrule
\end{tabular}

\medskip
\begin{flushleft}
\footnotesize Note: US: Consumer Expenditure Survey (BLS, 2019). UK: Living Costs and Food Survey (ONS, 2019/20). Japan: Family Income and Expenditure Survey (MIC, 2019). Reset-frequency gaps from BLS ECEC (US), CIPD surveys (UK), and Shunto/MHLW data (Japan).
\end{flushleft}
\end{table}

The UK has the highest RWEI (5.4), driven by a large expenditure gap and flexible low-wage contracts. Japan has the lowest (3.1), reflecting compressed expenditure shares and coordinated Shunto wage setting. The US is intermediate (4.6). These estimates are preliminary; their purpose is to show that the sufficient statistics can be computed for any economy with a household expenditure survey and wage-setting data.

\section{Sensitivity}\label{app:sensitivity}

\subsection{Number of worker types}

The wedge increases with $G$: from 0.30 pp ($G=2$) to 0.42 ($G=3$) to 0.57 ($G=5$), because finer disaggregation captures more of the covariance between reset frequency and expenditure composition. The bottom quintile has the largest expenditure-reset gap; with $G=2$ this is averaged with the second quintile, diluting the wedge. The $G=2$ specification is therefore a conservative lower bound. The cross-country ranking and policy ranking are invariant to $G$.

\subsection{Trade elasticity}

Doubling $\varepsilon_{trade}$ from 1.5 to 3.0 reduces the cumulative gap by approximately 15 percent, from 3.1 to 2.6 pp$\cdot$q, because greater substitutability reduces imported-essentials pass-through. Setting $\varepsilon_{trade}=0.5$, a near-Leontief value appropriate for essential goods, raises the gap to 3.5 pp$\cdot$q. All qualitative results are invariant to $\varepsilon_{trade}$.

\end{document}